\documentclass[11pt]{article}    
\usepackage[utf8]{inputenc}                

\usepackage[margin=1in]{geometry} 
\usepackage{mathtools} 
\usepackage{amsthm,amssymb}        
\usepackage{graphicx}
\usepackage{epstopdf}
\usepackage{array}
\usepackage{xcolor}
\usepackage{breqn}
\usepackage{enumitem}
\usepackage{multirow} 
\usepackage{longtable}
\usepackage{booktabs}
\usepackage[hidelinks]{hyperref}
\usepackage[tracking]{microtype}
\usepackage{natbib}
\title{Enumeration of coalescent histories for caterpillar species trees and $p$-pseudocaterpillar gene trees}
\author{Egor Alimpiev, Noah A Rosenberg \\ \bigskip \footnotesize{Department of Biology, Stanford University, Stanford, CA 94305 USA}}
\date{\today}                                          

\newtheorem{theorem}{Theorem}
\newtheorem{proposition}[theorem]{Proposition}

\newtheorem{lemma}[theorem]{Lemma}

\newtheorem{definition}[theorem]{Definition}

\newcommand{\doubleentry}[2]{{\renewcommand{\arraystretch}{1}$\begin{array}{c}#1\\#2\end{array}$}}



\begin{document}

\maketitle

\begin{abstract}
For a fixed set $X$ containing $n$ taxon labels, an ordered pair consisting of a gene tree topology $G$ and a species tree $S$ bijectively labeled with the labels of $X$ possesses a set of coalescent histories---mappings from the set of internal nodes of $G$ to the set of edges of $S$ describing possible lists of edges in $S$ on which the coalescences in $G$ take place. Enumerations of coalescent histories for gene trees and species trees have produced suggestive results regarding the pairs $(G,S)$ that, for a fixed $n$, have the largest number of coalescent histories. We define a class of 2-cherry binary tree topologies that we term \emph{$p$-pseudocaterpillars}, examining coalescent histories for non-matching pairs $(G,S)$, in the case in which $S$ has a caterpillar shape and $G$ has a $p$-pseudocaterpillar shape. Using a construction that associates coalescent histories for $(G,S)$ with a class of “roadblocked” monotonic paths, we identify the $p$-pseudocaterpillar labeled gene tree topology that, for a fixed caterpillar labeled species tree topology, gives rise to the largest number of coalescent histories. The shape that maximizes the number of coalescent histories places the “second” cherry of the $p$-pseudocaterpillar equidistantly from the root of the “first” cherry and from the tree root. A symmetry in the numbers of coalescent histories for $p$-pseudocaterpillar gene trees and caterpillar species trees is seen to exist around the maximizing value of the parameter $p$. The results provide insight into the factors that influence the number of coalescent histories possible for a given gene tree and species tree. \\


\end{abstract}


%

\section{Introduction}
\label{secIntroduction}

In mathematical phylogenetics, a coalescent history represents the paired list of coalescences in a gene tree together with their associated edges of a species tree. Consider two binary, rooted, leaf-labeled trees, $G$ and $S$, with leaves labeled by the same label set $X$, such that each label in $X$ is associated with exactly one leaf of $G$ and exactly one leaf of $S$. We regard $G$ as a \emph{gene tree} representing the evolution of genealogical lineages in a group of species, and $S$ as the \emph{species tree} representing the evolutionary descent of the species themselves.

For a gene tree $G$ evolving on a species tree $S$, a \emph{coalescent history} is a mapping from the set of internal nodes of $G$ to the set of internal edges of $S$, such that two rules are followed: (i) the image of an internal node $v$ of $G$ is ancestral in $S$ to each leaf of $S$ that shares a label with some leaf descended from $v$ in $G$; (ii) the image of an internal node $v$ of $G$ is ancestral in $S$ to the images of each of its descendant nodes. The biological interpretation of (i) is that a set of gene lineages can only find a common ancestor on a species tree edge that it is possible for them all to reach; the interpretation of (ii) is that the gene lineages descended from a descendant node coalesce at least as recently as do the gene lineages descended from its ancestral nodes. Note that we regard a node as trivially ancestral to and descended from itself. The coalescent histories for $(G,S)$ can be viewed as describing a discrete class of evolutionary scenarios for the lineages of $G$ on the edges of $S$.

A variety of studies have enumerated the coalescent histories for pairs $(G,S)$, both by a recursive approach that applies in all cases of $(G,S)$ \citep{rosenberg2007counting, ThanRuths07-JCB}, and by closed-form formulas and bijective constructions developed for particular families of trees \citep{degnan2005, rosenberg2007counting, rosenberg2013caterpillar, rosenberg2019enumeration, rosenberg2010coalescent, disanto2015coalescent, disanto2016asymptotic, himwich2020}. These enumeration studies, primarily considering \emph{matching} gene trees and species trees with $G=S$ and having particular emphasis on shapes such as \emph{caterpillars}, \emph{4-pseudocaterpillars}, and \emph{caterpillar-like} families (Figure~\ref{fig:tree_types}), have informally observed that in specified classes of trees, the largest number of coalescent histories tends to occur when the pair $(G,S)$ possesses two features: multiple different sequences exist in which the coalescences of $G$ can occur, and many edges of $S$ exist on which those coalescences can take place. 

\begin{figure}[t]\centering
\hspace{1cm}\begin{minipage}[b]{0.3\linewidth} (A) 

	{\centering\includegraphics[height=0.68in]{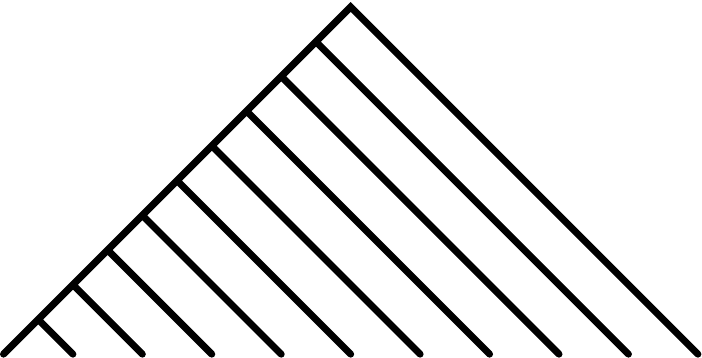}}
\end{minipage}
\begin{minipage}[b]{0.3\linewidth} (B) 

	{\centering\includegraphics[height=0.68in]{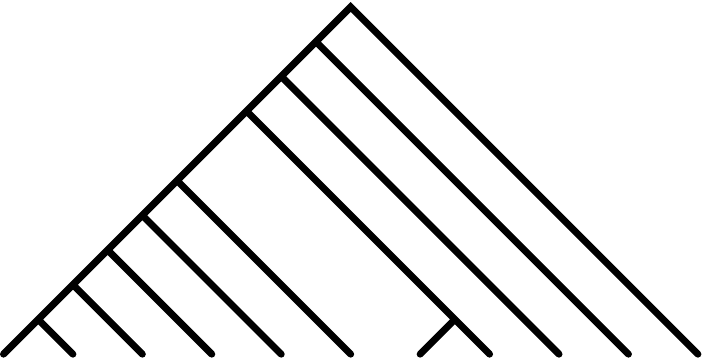}}
\end{minipage}
\begin{minipage}[b]{0.3\linewidth} (C) 

	{\centering\includegraphics[height=0.68in]{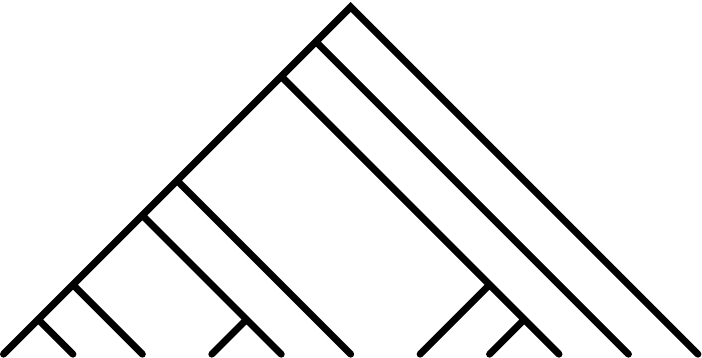}}
\end{minipage}
\caption{Three tree shapes with $n=11$ leaves. (A) Caterpillar tree shape. (B) $p$-pseudocaterpillar tree shape. For this tree, $p=8$. (C) Caterpillar-like tree shape. The seed tree has size 9.}
\label{fig:tree_types}
\end{figure}

\cite{rosenberg2007counting} observed that for small trees with at most $n=9$ leaves and $G=S$, the largest numbers of coalescent histories among tree pairs $(G,G)$ with fixed $n$ were seen for trees that had structure similar to caterpillar trees, but that unlike caterpillars, had more than one possible sequence of coalescences. \cite{rosenberg2013caterpillar} and \cite{disanto2016asymptotic} examined tree families $(G,G)$, with $n$ growing arbitrarily large in specified caterpillar-like tree families. Beginning with a seed tree, these studies generated families of increasingly large trees by sequentially adding taxa so that the next tree in a family was formed by placing the current tree and a single leaf on opposite sides of a new root. They saw that across all seed trees of a fixed small size, as the number of leaves grew without bound, the largest numbers of coalescent histories occurred when the seed tree had many different sequences in which its coalescences could take place. \cite{disanto2015coalescent} constructed a tree family, the \emph{lodgepole} family, that, unlike caterpillar families, grows so that as the number of leaves increases, trees accumulate both new sequences in which coalescences can take place and new places for them to occur. This family is the family of matching tree pairs $(G,G)$ with the largest-known number of coalescent histories as $n$ increases without bound.

Despite many observations suggesting that coalescent histories tend to increase in number when $G$ has many sequences in which coalescences can take place and many edges on which those coalescences can occur, existing results in support of this view have focused on small trees \citep{rosenberg2007counting} and on informal interpretations of specific families with large limits as $n \rightarrow \infty$ \citep{rosenberg2013caterpillar,  disanto2016asymptotic}; no result has formally demonstrated the observation in a class of trees for a fixed finite $n$ of arbitrary size. We devise a scenario to formalize this idea characterizing scenarios with the largest numbers of coalescent histories. We fix the species tree $S$ to be a caterpillar, and we consider a family of non-matching gene trees $G$, the \emph{$p$-pseudocaterpillars}. We show that among this class of non-matching pairs $(G,S)$ with fixed $n$, the largest number of coalescent histories occurs precisely when $G$ combines these two elements: many coalescence sequences, and many edges on which those coalescences can take place.

Our approach is one of relatively few to examine enumerations of coalescent histories in the case that $G$ is not necessarily equal to $S$ \citep{rosenberg2007counting, rosenberg2019enumeration, ThanRuths07-JCB, rosenberg2010coalescent, himwich2020}. The strategy employs a construction that enumerates coalescent histories for non-matching caterpillar trees. Generalizing a result of \cite{degnan2005} for enumeration of coalescent histories for matching caterpillar trees, \cite{himwich2020} produced a bijection with monotonic paths for use in enumerating coalescent histories for non-matching caterpillar pairs $(G,S)$. We use this monotonic-path construction to enumerate coalescent histories for the class of non-matching trees that considers a caterpillar species tree $S$ and a $p$-pseudocaterpillar gene tree $G$. 

Section~\ref{secPreliminaries} introduces definitions and notation. Section~\ref{secExample} gives an example that motivates the general calculation. In Section~\ref{sec:general_construction}, we  enumerate coalescent histories in the general case. Section~\ref{secSpecialCases} gives special cases with specified values of $p$. Finally, in Section~\ref{sec:max}, for a specified caterpillar species tree $S$ of fixed size $n$, considering all possible values of $p$, we obtain the maximal number of coalescent histories across all non-matching $p$-pseudocaterpillar trees $G$. Section~\ref{secSymmetry} discusses a symmetry in $p$ for fixed $n$, and we conclude with a discussion in Section~\ref{secDiscussion}. The computations illustrate how the monotonic path approach of \cite{himwich2020} in the case of caterpillar species trees can be extended to enumerate coalescent histories in more cases beyond that of caterpillar gene trees. 

\section{Preliminaries}
\label{secPreliminaries}

We formally define coalescent histories and $p$-pseudocaterpillars in Sections~\ref{secCoalescentHistories} and~\ref{secPseudocaterpillars}, and we introduce results concerning the Catalan numbers in Section~\ref{secCatalan}. In Section~\ref{secHimwich}, we describe the use of monotonic paths to enumerate coalescent histories for caterpillar tree pairs.

\subsection{Coalescent histories}
\label{secCoalescentHistories}

The definitions in this article closely follow \cite{himwich2020}. Henceforth, we treat all ``trees'' as binary, rooted, and leaf-labeled, except where specified. The set of vertices or nodes of a tree can be divided into \emph{leaf nodes} and non-leaf \emph{internal} nodes. For rooted tree $G$, we say that a node $v_1$ is \emph{descended} from a node $v_2$ if the shortest path from $v_1$ to the root of $G$ travels through $v_2$; $v_2$ is then \emph{ancestral} to $v_1$. Ancestor--descendant relationships also apply to edge--edge pairs and edge--node pairs. A node or edge is trivially descended from and ancestral to itself. Each internal node, including the root, possesses an associated \emph{internal edge} immediately ancestral to it.

We consider pairs $(G,S)$ in which $G$ represents a gene tree, describing the descent of a set of genealogical lineages, and $S$ represents a species tree, describing the descent of a set of species. $G$ and $S$ are assumed to have the same number of leaves, $n$. We assume that the leaf set of $G$ and the leaf set of $S$ are labeled by the same label set $X$, and that each label in $X$ is assigned to exactly one leaf of $G$ and to exactly one leaf of $S$. This assumption corresponds to an assumption that exactly one gene lineage is sampled in each of the $n$ species.

For the pair $(G,S)$, we can formally define the functions known as  \emph{coalescent histories}.
\begin{definition}
\label{defCoalescentHistories}
	Consider a pair of trees $(G,S)$ that are binary, rooted, and leaf-labeled, with the labels in bijective correspondence. A \emph{coalescent history} is a function $\alpha$ from the set of internal nodes of $G$ to the set of internal edges of $S$, satisfying two conditions:
		\begin{enumerate}[label=(\alph*)]
		\item For each internal node $v$ in $G$, all labels for leaves descended from $v$ in $G$ label leaves descended from edge $\alpha(V)$ in $S$.
		\item For each pair of internal nodes $v_1$ and $v_2$ in $G$, if $v_2$ is descended from $v_1$, then $\alpha(v_2)$ is descended from $\alpha(v_1)$ in $S$.
	\end{enumerate}
\end{definition}

\noindent In this definition, nodes of $G$ represent coalescent events for the gene lineages, and edges of $S$ represent species tree edges along which the gene lineages evolve. A coalescent history reflects the biological process of coalescence, in which descendants cannot coalesce farther back in time than their ancestors. Ancestor--descendant relations are preserved under the mapping $\alpha$.


\begin{figure}[tb]\centering
\hspace{2cm}\begin{minipage}[b]{0.4\linewidth} (A)

	{\centering\includegraphics[height=1.08in]{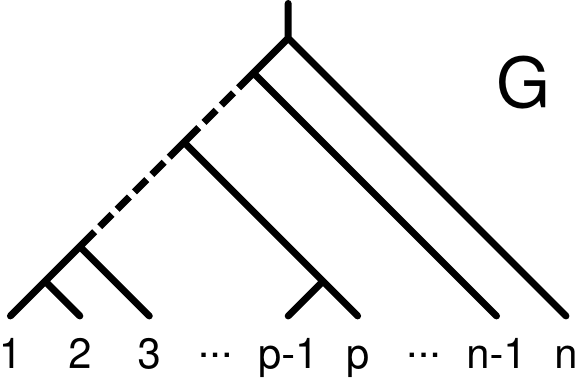}}
\end{minipage}
\begin{minipage}[b]{0.4\linewidth} (B)

	{\centering\includegraphics[height=0.92in]{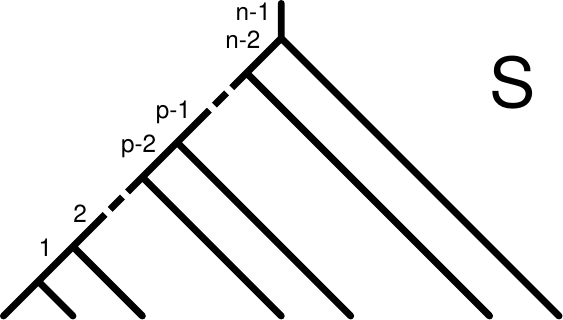}}	
	\vspace{3.6mm}
\end{minipage}
\caption{Labeled $p$-pseudocaterpillar gene tree and caterpillar species tree. (A) A $p$-pseudocaterpillar gene tree $G$ with coordinates for the leaves. (B) A species tree $S$ with the internal edges labeled. Both trees have $n$ leaves. The trees are drawn in \emph{canonical form}, so that the shortest path from the left-most leaf to the root contains all other internal nodes for the caterpillar, and all other internal nodes except one for the $p$-pseudocaterpillar.}
\label{fig:coord}
\end{figure}

\subsection{Caterpillars and $p$-pseudocaterpillars}
\label{secPseudocaterpillars}

As our goal is to enumerate the coalescent histories in the case that $G$ has a $p$-pseudocaterpillar topology and $S$ has a caterpillar topology, we define caterpillar and $p$-pseudocaterpillar shapes for binary, rooted tree topologies. 

\begin{definition}
\label{caterpillar_tree}
	A \emph{caterpillar tree} is a binary, rooted tree that has an internal node that is descended from all other internal nodes. 
\end{definition}
\noindent In a caterpillar tree, each internal node has at least one leaf as an immediate descendant (Figure~\ref{fig:tree_types}A). Equivalently, a caterpillar tree is a tree that has only one \emph{cherry node}: an internal node with exactly two descendant leaves. 

\cite{rosenberg2007counting} defined binary, rooted pseudocaterpillar trees with $n \geq 4$ leaves as trees in which all internal nodes except one have at least one immediate leaf descendant. The node that provides the exception has two cherry nodes as its immediate descendants. We generalize the earlier definition of  pseudocaterpillar trees to consider \emph{generalized pseudocaterpillar trees}. To define this concept, we denote by $v_L$ and $v_R$ the left and right descendant nodes of an internal node $v$.
\begin{definition}
	\label{pseudocaterpillar}
	A \emph{generalized pseudocaterpillar tree} is a binary, rooted tree that has at least four leaves and that satisfies two conditions. (i) The tree possesses exactly two cherry nodes. (ii) For each internal node $v$, at least one of ${v_L,v_R}$ has no more than two descendant leaves.
\end{definition}
\noindent In other words, a generalized pseudocaterpillar tree is formed from a caterpillar tree, replacing one of the leaves not descended from the unique cherry node by a second cherry node (Figure~\ref{fig:tree_types}B). 

A generalized pseudocaterpillar can be described by two numbers: the total number of leaves $n$ and the position $p$ of the ``second'' cherry. To precisely identify $p$ for a generalized pseudocaterpillar tree, we label the leaves by natural numbers starting from left to right, placing the ``first'' cherry---the one present in the caterpillar from which the generalized pseudocaterpillar has been generated---on the left. We define the position of the second cherry as the number corresponding to its second leaf from the left (Figure~\ref{fig:coord}A). A generalized pseudocaterpillar tree with a second cherry in position $p$, $4 \leq p \leq n$, is termed a \emph{$p$-pseudocaterpillar tree}. The pseudocaterpillar trees in the sense of \cite{rosenberg2007counting} are 4-pseudocaterpillars.

Our interest is in the case in which the gene tree has $p$-pseudocaterpillar topology for some $p$, and the species tree has a caterpillar topology. A caterpillar species tree with $n$ leaves has  $n-1$ edges on which gene tree coalescences can happen; we also label these edges with natural numbers, following the order from \cite{degnan2005gene} (Figure~\ref{fig:coord}B).

\subsection{Catalan numbers}
\label{secCatalan}

It is useful to introduce the Catalan number sequence $1,1,2,5,14,42,132,429, \ldots$, as it features prominently in our analysis. Letting $\mathcal{C}_n$ be the $n$th Catalan number for $n \geq 0$, 
\begin{equation}
    \label{catalan}
    \mathcal{C}_n = \binom{2n}{n} - \binom{2n}{n-1} = \frac{1}{n+1}\binom{2n}{n}.
\end{equation}
Considering the many combinatorial interpretations of this sequence~\citep{concretemath, stanley_catalan}, we will make use of the fact that $\mathcal{C}_n$ is the number of monotonic paths that travel from $(0,0)$ to $(n,n)$ on a square lattice of size $n\times n$ and that do not cross the diagonal connecting $(0,0)$ to $(n,n)$, where a monotonic path is a path that proceeds exclusively by 1-unit steps up or to the right.

A \emph{Catalan triangle} is a combinatorial structure that counts monotonic paths to points on the lattice that lie on or below the diagonal \citep{reuveni2014catalan}. Entry $(n,k)$ of the Catalan triangle gives the number of monotonic paths on the square lattice that travel from the origin to a point $(n,k)$ and that do not cross the $y=x$ line. The number of such paths, which have $n$ ``right-steps'' and $k$ ``up-steps,'' is \citep{reuveni2014catalan}:
	\begin{equation}
   	\label{catalan_entry}
   		C(n,k) = 
   		\begin{cases}
    		\binom{n+k}{k} - \binom{n+k}{k-1} & 1\leq k \leq n\\
    		1 & k=0\\
    		0 & k > n.  
    	\end{cases}
	\end{equation}
For $n \geq 1$ and $1 \leq k \leq n$, this function satisfies the first-order recurrence 
	\begin{equation*}
		C(n,k) = C(n-1,k)+C(n,k-1).
	\end{equation*}
	
A \emph{Catalan trapezoid} is obtained in a similar way, except that we allow additional $m-1$ up-steps to happen starting at the origin, so that monotonic paths that do not travel above the diagonal from $(0,m-1)$ to $(n,n+m-1)$ are tabulated. The number $m$ is called the \emph{order} of the trapezoid; $m=1$ corresponds to the Catalan triangle. Entry $(n,k)$ of the Catalan trapezoid of order $m$ is given by
	\begin{equation}
		\label{trapezoid_entry}
		C_t(n,k,m) = 
   		\begin{cases}
    		\binom{n+k}{k} & 0 \leq k \leq m-1 \\
    		\binom{n+k}{k} - \binom{n+k}{k-m} & m\leq k \leq n+m-1 \\
    		0 & k>n+m-1,
    	\end{cases}
	\end{equation}
and it satisfies a similar recurrence $C_t(n,k,m)=C_t(n-1,k,m)+C_t(n,k-1,m)$ for $n \geq 1$ and $1 \leq k \leq n+m-1$. With the origin in the lower left corner, the first colums of the Catalan triangle and the Catalan trapezoid of order 3 appear below:
	\begin{align*}
	\label{cat_entries}
		&\begin{array}{ccccccc}
 			&   &   &   &    &    & 132 \\
  			&   &   &   &    & 42 & 132 \\
  			&   &   &   & 14 & 42 &  90 \\
  			&   &   & 5 & 14 & 28 &  48 \\
  			&   & 2 & 5 &  9 & 14 &  20 \\
  			& 1 & 2 & 3 &  4 &  5 &   6 \\
 		  1 & 1 & 1 & 1 &  1 &  1 &   1 \\
		\end{array}
		&\begin{array}{ccccccc}
 			  &   &   &    & 90 \\
  			  &   &   & 28 & 90 \\
  			  &   & 9 & 28 & 62 \\
  			  & 3 & 9 & 19 & 34 \\
  			1 & 3 & 6 & 10 & 15 \\
 			1 & 2 & 3 &  4 &  5 \\
 			1 & 1 & 1 &  1 &  1 \\
		\end{array}
	\end{align*}

\subsection{Bijection between coalescent histories and monotonic paths for caterpillars}
\label{secHimwich}

Building on work of \cite{degnan2005}, the bijective construction of \cite{himwich2020} enumerates coalescent histories for pairs consisting of a caterpillar gene tree and a caterpillar species tree by bijectively associating each coalescent history with a monotonic path that does not cross the diagonal of a square lattice. The coalescent histories are then enumerated by counting the bijectively-associated monotonic paths.

In the construction, given a caterpillar species tree $S$ and a caterpillar gene tree $G$ with $n$ leaves, a square $(n-1) \times (n-1)$ lattice is examined. The coalescent histories for $(G,S)$ correspond to monotonic paths from $(0,0)$ to $(n-1,n-1)$, with each right-step corresponding to a species tree internal edge, and each up-step corresponding to a gene tree coalescence. The pair $(G,S)$ specifies a set of \emph{roadblocks}, points in the lattice through which monotonic paths are not permitted to travel. The number of coalescent histories for $(G,S)$ then equals the number of monotonic paths that do not cross the diagonal and that do not travel through any of the roadblocks. In the case that $G$ and $S$ have the same caterpillar labeled topology, no roadblocks exist, and the number of monotonic paths, and hence the number of coalescent histories, is the Catalan number $\mathcal{C}_{n-1}$ \citep{degnan2005}.
 
The construction of \cite{himwich2020} also applies to caterpillar subtrees. Suppose $G$ possibly has fewer leaves than $S$, so that the label set for $G$ is a subset of the label set for $S$. If we use the term \emph{partial coalescent history} to describe mappings that satisfy Definition~\ref{defCoalescentHistories} except that the label set of $G$ is a subset of the label set of $S$ rather than a bijectively-associated label set, then the number of partial coalescent histories for a caterpillar pair $(G,S)$ is obtained by counting roadblocked monotonic paths to an associated point that is not necessarily the point $(n-1,n-1)$. 

For details, see \cite{himwich2020}. We illustrate the construction in an example.
\section{Example}
\label{secExample}

\begin{figure}[t]\centering
	\hspace{1.6cm}\begin{minipage}[b]{0.4\linewidth} (A)
	
	{\centering
	\includegraphics[width=0.7\textwidth]{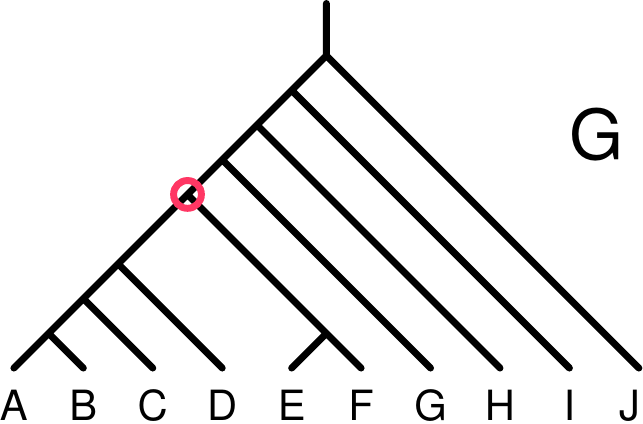}}
\end{minipage}
\begin{minipage}[b]{0.4\linewidth} (B)
	
	{\centering
	\includegraphics[width=0.7\textwidth]{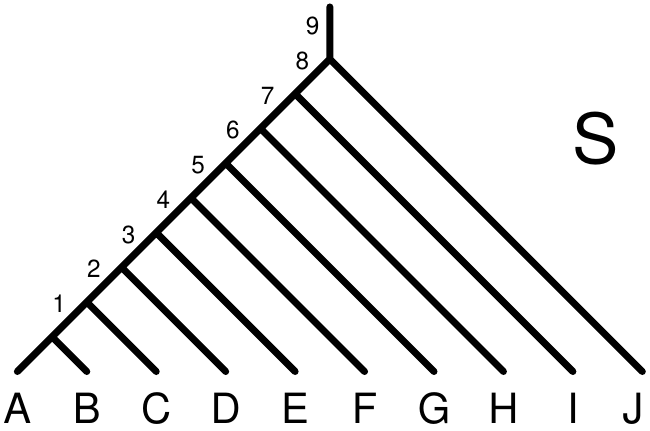}}
\end{minipage}
\caption{Example (gene tree, species tree) pair. (A) 6-pseudocaterpillar gene tree $G$ with ``second cherry'' (E,F). The pivotal coalescence is circled in red. (B) Caterpillar species tree $S$.}
\label{fig:ex}
\end{figure}

\begin{figure}[t]\centering
	\hspace{2cm} \begin{minipage}[b]{0.4\linewidth} (A) \vspace{7.5mm}
	
	{\centering \includegraphics[height=1.2in]{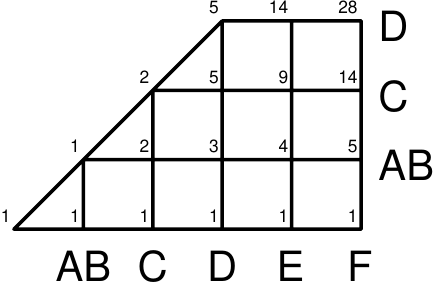}}
\end{minipage}
\begin{minipage}[b]{0.4\linewidth} (B)
	
	{\centering \includegraphics[height=1.52in]{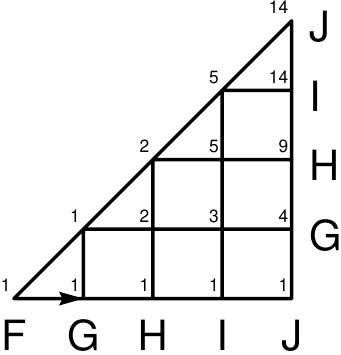}}
\end{minipage}
\caption{Catalan triangle construction for enumerating coalescent histories for Figure~\ref{fig:ex}. Following \cite{himwich2020}, up-steps represent gene tree coalescences and right-steps represent species tree edges. The numbers indicated represent counts of monotonic paths according to eq.~\eqref{catalan_entry}. (A) Diagram corresponding to the left subtree descended from the pivotal gene tree coalescence. (B) Diagram corresponding to the portion of the gene tree ancestral to the pivotal coalescence. The arrow indicates that all coalescences other than those depicted in the diagram have already happened by the starting point, the first species tree internal edge ancestral to F.}
\label{fig:diag_ex}
\end{figure}

Our approach to extending the construction of \cite{himwich2020} to count coalescent histories for a caterpillar species tree and a non-matching $p$-pseudocaterpillar gene tree---a tree with one extra cherry---can be understood with an example. Consider a gene tree $G$ with 10 leaves, with cherry node (E,F) as shown in Figure~\ref{fig:ex}A, and a species tree $S$ as shown in Figure~\ref{fig:ex}B.

The key to counting coalescent histories for $(G,S)$ is to examine the specific gene tree coalescence circled in red in Figure~\ref{fig:ex}A, indicating the most recent common ancestor of both cherries of $G$. We call this node the \emph{pivotal} coalescence. In a coalescent history, this pivotal coalescence can take place on any internal edge ancestral to species F in the species tree. We partition all coalescent histories for $(G,S)$ by the position of this pivotal coalescence. For each placement of the pivotal coalescence, we then count the number of coalescent histories by counting monotonic paths on particular diagrams for the subtrees generated by the pivotal coalescence. 

Suppose the pivotal coalescence of $G$ happens on edge 5 of $S$. Then all the coalescences in the ``left'' subtree descended from the pivotal coalescence must happen on or before edge 5. This left subtree is now a caterpillar (((A,B),C),D), coalescing on a caterpillar (((((A,B),C),D),E),F). We can now follow the construction of \cite{himwich2020} to enumerate partial coalescent histories through a bijection with monotonic paths.

In particular, the number of ways that the gene tree coalescences of (((A,B),C),D) can occur on species tree (((((A,B),C),D),E),F) is equal to the number of monotonic paths on a Catalan triangle restricted to 5 right-steps and 3 up-steps (Figure~\ref{fig:diag_ex}A). The up-steps correspond to the 3 coalescences in the subtree (((A,B),C),D), and the right-steps correspond to the 5 edges of  (((((A,B),C),D),E),F) on which they can take place. Following eq.~\eqref{catalan_entry}, the number of monotonic paths that travel from $(0,0)$ to $(5,3)$ and that do not cross the diagonal is ${\binom83} - {\binom82} = 28$. Because gene tree coalescence (E,F) must occur on species tree edge 5 when the pivotal coalescence occurs on edge 5, coalescence (E,F) does not introduce additional coalescent histories. Thus, 28 possible partial coalescent histories place the pivotal coalescence on species tree edge 5.

We now need to consider the coalescences ancestral to the pivotal coalescence. Coalescences involving leaf G can happen on edge 6 or on any edge ancestral to 6, coalescences with leaf H can happen on edge 7 or any edge ancestral to 7, provided that leaf G has already participated in a coalescence, and so on. Again following the construction of \cite{himwich2020}, the possible assignments of gene tree coalescences to species tree edges in this upper part of the species tree can be described by a Catalan triangle with 4 right-steps and 4 up-steps (Figure~\ref{fig:diag_ex}B). There are 14 possible monotonic paths.

To obtain the total number of coalescent histories with pivotal coalescence on edge 5, we now multiply the two numbers we already have: for each of the 28 partial coalescent histories for coalescences descended from the pivotal coalescence, there are 14 ways for the coalescences ancestral to it to happen. Hence, 392 coalescent histories exist with pivotal coalescence on edge 5. 

To obtain the total count of coalescent histories for $(G,S)$, we must consider all other possible locations of the pivotal coalescence, and sum their associated numbers of coalescent histories. With this idea, however, we are now ready for the general case.

\section{General construction}
\label{sec:general_construction}

The example in Section \ref{secExample} illustrates that we can enumerate coalescent histories for a $p$-pseudocaterpillar gene tree and a caterpillar species tree by dividing the problem into three components: placement of the pivotal coalescence, and two enumerations, one for coalescences descended from the pivotal coalescence, and the other for coalescences ancestal to it. We describe these two enumerations in full generality, and complete the calculation by summing over all placements of the pivotal coalescence.

Consider a $p$-pseudocaterpillar gene tree $G$ and a caterpillar species tree $S$, both with $n$ leaves, and bijectively labeled with the same set of distinct labels. Suppose $G$ and $S$ have an \emph{identical} leaf labelling, by which we mean that when $G$ and $S$ are drawn in canonical form (Figure \ref{fig:coord}), the gene tree and species tree labels are listed in the same order when reading them from left to right. Figure \ref{fig:ex} illustrates an identical leaf labeling. Note that labelings in which the labels in one or both cherries of $G$ are transposed with respect to $S$ also qualify as identical.

Using our numerical labeling scheme for edges of gene trees and species trees (Figure \ref{fig:coord}), the pivotal coalescence can take place on any species tree edge from $p-1$ to $n-1$. Suppose it happens on edge $k$, $p-1 \leq k \leq n-1$.

\subsection{Coalescences descended from the pivotal coalescence}
\label{secPivotal}

Label by $S_k$ the subtree of $S$ whose root is the node immediately descended from edge $k$. Label the subtree of $G$ whose root node is the pivotal coalescence by $G_*$. The left subtree of $G_*$, which we label $G_{*\ell}$, is a caterpillar with $p-3$ coalescences. By the assumption that the pivotal coalescence takes place on species tree edge $k$, all coalescences in $G_{*\ell}$ must occur on edges $1,2,\ldots,k$.

Following \cite{himwich2020}, the partial coalescent histories for $(G_{*\ell}, S_k)$, with $p-3$ gene tree coalescences and $k$ species tree edges on which they take place, correspond to monotonic paths from $(0,0)$ to $(k,p-3)$ that do not cross the $y=x$ line. The number of partial coalescent histories therefore corresponds to Catalan triangle entry $(k,p-3)$. By eq.~\eqref{catalan_entry}, this quantity, which we denote $\ell_k$, equals
\begin{equation}
\label{lk}
\ell_k = \binom{k+p-3}{p-3}-\binom{k+p-3}{p-4}.
\end{equation}

The right subtree of $G_*$, or $G_{*r}$, has exactly one coalescence, which can happen on any of the branches $p-1,p,\ldots, k$. Hence, the number of coalescent histories for $(G_{*r}, S_k)$ is 
\begin{equation}
\label{rk}
r_k = k-p+2.
\end{equation}
Combining the left and right subtrees of $G_*$, from eqs.~\eqref{lk} and \eqref{rk}, the number of partial coalescent histories for $(G_*,S_k)$ is $\ell_k r_k$.

\subsection{Coalescences ancestral to the pivotal coalescence}
\label{secUpperPart}

To examine coalescences ancestral to the pivotal coalescence, the pivotal coalescence can be viewed as a ``leaf'' of a caterpillar gene tree $G^*$ whose coalescences occur on species tree edges numbered $k$ or greater. In this view, $G^*$ is the $(n-p+1)$-leaf caterpillar tree in which the subtree rooted at the pivotal coalescence is replaced by a leaf, so that the pivotal coalescence is a leaf in the cherry of $G^*$.

$G^*$ has $n-p$ gene tree coalescences, which take place on species tree edges $k,k+1,\ldots,n-1$, a total of $n-k$ edges. It is possible for multiple coalescences in $G^*$ to occur on branch $k$; taking into account that branch $k$ has $k+1$ descendant leaves, and $p-1$ coalescences have already occured including the pivotal coalescence, at most $k-p+1$ coalescences of $G^*$ can occur on branch $k$. 

\begin{figure}[tb]
\centering
\begin{minipage}[b]{0.5\linewidth} (A)

\includegraphics[width=\linewidth]{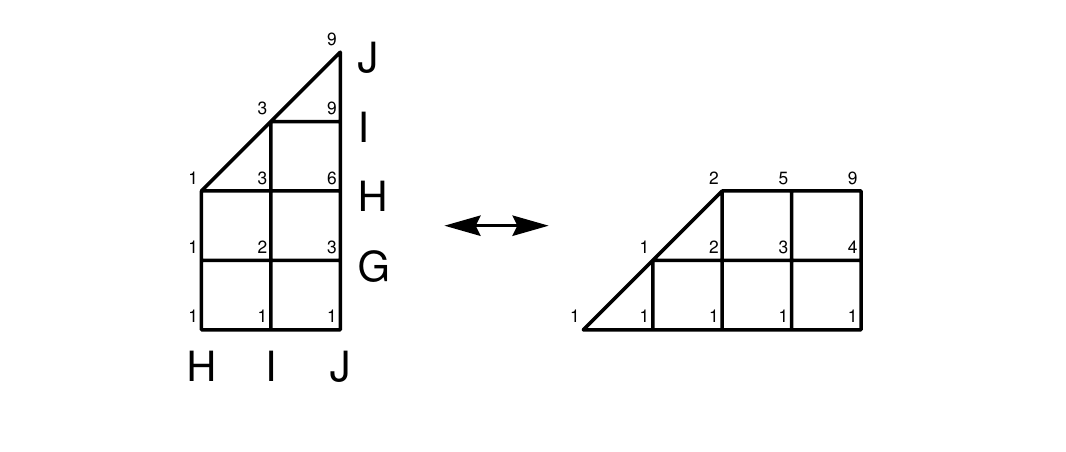}
\end{minipage}
\begin{minipage}[b]{0.5\linewidth} (B)

\includegraphics[width=\linewidth]{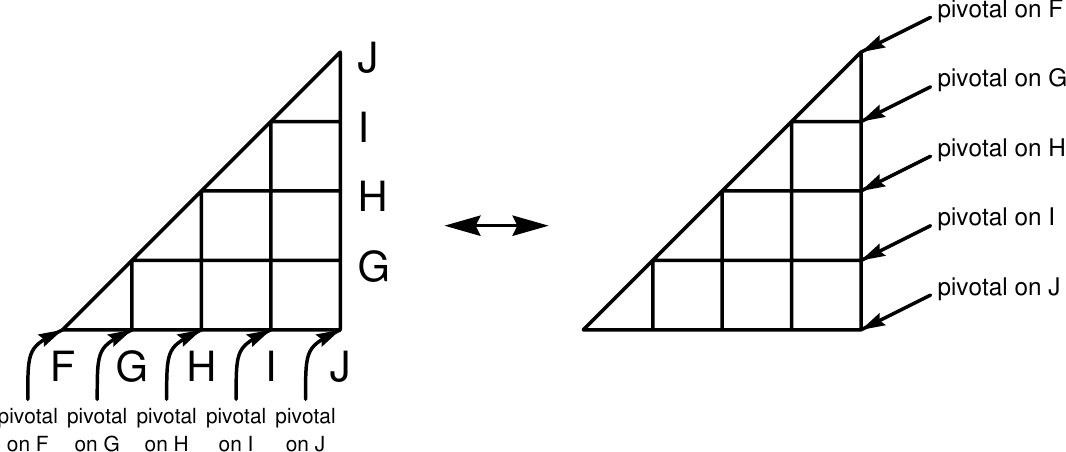}
\end{minipage}
\caption{Monotonic path construction for coalescences ancestral to the pivotal coalescence. (A) Exchanging the starting and ending points of monotonic paths, the number of monotonic paths in a trapezoidal lattice is equal to an entry in a Catalan triangle. Following the notation of Figure~\ref{fig:ex} with $n=10$ and $p=6$, suppose the pivotal coalescence happens on species tree edge $k=7$. Up to two additional coalescences can happen on edge $k=7$, producing a trapezoid. The number of monotonic paths is obtained from a Catalan trapezoid of order $k-p+2=3$ with $n-k-1=2$ right-steps and $n-p=4$ up-steps, or by symmetry, by computing entry $(n-p,n-k-1)=(4,2)$ of a Catalan triangle. (B) Monotonic path construction for each of the $n-p+1=5$ options for placement of the pivotal coalescence. For a placement of the pivotal coalescence shown in the left-hand diagram, the number of coalescent histories for $(G^*,S)$ is obtained by counting monotonic paths in the right-hand diagram from the lower-left vertex to an associated point on the right-hand edge.}
\label{fig:upper}
\end{figure}

The coalescences of $G^*$ therefore correspond to monotonic paths that do not cross a specified diagonal of a trapezoidal lattice. The number of right-steps is $n-k-1$, one for each non-root edge on which coalescences take place, and the number of up-steps is $n-p$, one for each gene tree coalescence in $G^*$. The order of the trapezoid is $k-p+2$, one more than the number of coalescences of $G^*$ that can occur on the initial branch $k$. 

We can count these monotonic paths using eq.~\eqref{trapezoid_entry}, or by noting that the number of monotonic paths is symmetric with respect to interchange of the starting and ending points. By symmetry, the number of paths on a Catalan trapezoid is then equal to one of the entries in the right-most column of some Catalan triangle. We have 
\begin{equation}
    C(n-p,n-k-1)=C_t(n-k-1,n-p,k-p+2).
\end{equation}
A visual explanation appears in Figure~\ref{fig:upper}.

Denote by $u_k$ the number of coalescent histories for $(G^*,S)$. Using eq.~\eqref{trapezoid_entry} or the symmetry argument with eq.~\eqref{catalan_entry} to count monotonic paths in a triangular lattice with $n-p$ \emph{right}-steps and $n-k-1$ \emph{up}-steps, we have
\begin{equation}
\label{uk}
u_k = C(n-p,n-k-1) = \binom{2n-p-k-1}{n-k-1}-\binom{2n-p-k-1}{n-k-2}.
\end{equation}

\subsection{Full formula}
\label{secFullFormula}

We have shown that the number of partial coalescent histories for $(G_*,S)$ that place the pivotal coalescence on edge $k$ is $\ell_k r_k$, and that for each of these partial coalescent histories, the number of partial coalescent histories for $(G^*,S)$ is $u_k$. Because each coalescent history for $(G,S)$ consists of a partial coalescent history for $(G_*,S)$, a placement of the pivotal coalescence, and a partial coalescent history for $(G^*,S)$, the number of coalescent histories for the case in which the pivotal coalescence happens on edge $k$ is $\ell_kr_ku_k$. Summing over values of $k$, we have proven the following theorem.

\begin{theorem}
\label{thm:coalescent}
Consider a caterpillar species tree $S$ with $n \geq 4$ leaves and an identically-labeled $p$-pseudocaterpillar gene tree $G$ with $n$ leaves and $4 \leq p \leq n$. With $C(n,k)$ as in eq.~\eqref{catalan_entry}, the number of coalescent histories for $(G,S)$ is
	\begin{equation}
    \label{coalescent}
    h(n,p) = \sum_{k=p-1}^{n-1}\ell_kr_ku_k = \sum_{k=p-1}^{n-1}C(k,p-3) \, (k-p+2) \, C(n-p,n-k-1).
    \end{equation}
\end{theorem} 
Note that eq.~\eqref{coalescent} can be seen to apply for $(n,p)$ with $p=3$ and $3 \leq p \leq n$, and hence for $n=3$. In this case, $G$ is viewed as a caterpillar gene tree whose cherry joins leaves 2 and 3. $G_{*\ell}$ has no coalescences, so $\ell_k = 1$; this enumeration accords with the definition of the function $C$ in eq.~\eqref{catalan_entry}, where we have $C(k,p-3)=C(k,0)=1$ for all $k$.

A convenient form of Eq.~\eqref{coalescent} for computation is as follows:
\begin{equation}
    \label{coalescent_exp}
    h(n,p) = \sum_{k=p-1}^{n-1} \frac{(k-p+2)^2(k-p+4)(2n-p-k-1)!(k+p-3)!}{(k+1)!(n-k-1)!(n-p+1)!(p-3)!}.
\end{equation}

\subsection{Identical and non-identical leaf labelings}
\label{secIdentical}

The results of \cite{himwich2020} enable a result on leaf labelings. We claim that for a fixed caterpillar species tree, an identically-labeled $p$-pseudocaterpillar gene tree---the focus of our analysis---has strictly more coalescent histories than any non-identically-labeled $p$-pseudocaterpillar. The argument is that any non-identically-labeled gene tree introduces at least one ``roadblock,'' decreasing its associated number of monotonic paths compared to the case of identical labels.

\begin{proposition}
	\label{permutations}
	Consider a caterpillar species tree $S$ with $n \geq 4$ leaves and a value of $p$, $4 \leq p \leq n$. The number of coalescent histories for $(G,S)$, with $G$ a $p$-pseudocaterpillar gene tree bijectively labeled with the same $n$ labels as $S$, is bounded above by $h(n,p)$, with equality if and only if $G$ and $S$ are identically labeled.
\end{proposition}
\begin{proof}
Theorem \ref{thm:coalescent} demonstrates that the number of coalescent histories is $h(n,p)$ in the identically-labeled case. We must show that a non-identically-labeled $G$ produces fewer coalescent histories.

Fix $n$ and $p$. Consider caterpillar species tree $S$ and $p$-pseudocaterpillar gene tree $G$, bijectively labeled with $n$ labels $\{A_1,A_2, \ldots, A_n\}$, but not necessarily identically labeled. Suppose that from left to right, $A_1, A_2, \ldots, A_n$ label the leaves of $S$ when $S$ appears in canonical form.

In eq.~\eqref{coalescent}, the form of the equation $h(n,p) = \sum_{k=p-1}^{n-1} \ell_k r_k u_k$ has a sum from $k=p-1$ to $n-1$ of a product of three quantities. Each quantity 
counts the number of monotonic paths on a Catalan triangle---trivially so in the case of $r_k=C(k-p+2,1)=k-p+2$, which represents the number of monotonic paths that proceed $k-p+2$ steps to the right and one step up.

If we now change $G$ to a possibly non-identically-labeled $p$-pseudocaterpillar $G'$, then the coalescent histories can be enumerated by a corresponding decomposition $h'(n,p) = \sum_{k=p-1}^{n-1} \ell_k' r_k' u_k'$, where, with the pivotal coalescence on edge species tree edge $k$,  $\ell_k'$, $r_k'$, and  $u_k'$ count partial coalescent histories for $(G_{* \ell}',S_k)$, $(G_{* r}',S_k)$, and $({G'}^{*},S)$, respectively.

To demonstrate that $h'(n,p) < h(n,p)$, we argue that $\ell_k' \leq \ell_k$, $r_k' \leq r_k$, and $u_k' \leq u_k$, and that for $G' \neq G$, at least one of these inequalities is strict. Following the argument of Corollary 11 of \cite{himwich2020}, the quantities $\ell_k'$,  $r_k'$, and $u_k'$ count monotonic paths that do not cross the $y=x$ line, that respectively proceed from $(0,0)$ to $(k,p-3)$, $(0,0)$ to $(k-p+2,1)$, and $(0,0)$ to $(n-p,n-k-1)$, possibly with roadblocks.

When $G^'=G$, no roadblocks occur, so that $\ell_k' \leq \ell_k$, $r_k' \leq r_k$, and $u_k' \leq u_k$. When $G' \neq G$, however, at least one of the following three statements holds: (i) $G_{* \ell}' \neq G_{* \ell}$; (ii) $G_{* r}' \neq G_{* r}$; (iii) ${G'}^{*} \neq G^*$. In the first case, for at least one $k$, a roadblock occurs in tabulating coalescent histories for $G_{*\ell}$, so that $\ell_k' < \ell_k$. Similarly, in the second case, for at least one $k$, a roadblock occurs in tabulating coalescent histories for $G_{*r}$, so that $r_k' < r_k$; in the third case, for at least one $k$, a roadblock occurs in tabulating coalescent histories for $G^*$, so that $u_k' < u_k$. 
\end{proof}

Note that for the sum describing the number of coalescent histories of $(G,S)$ to even proceed over the full range from $k=p-1$ to $n-1$, the first $p$ labels of $G$ from left to right when $G$ is written in canonical form must be a permutation of $A_1, A_2, \ldots, A_p$. Otherwise, at least one label of $G$ must be indexed by a value that exceeds $p$ and therefore cannot descend from edge $p-1$ of $S$.

\section{Small $p$}
\label{secSpecialCases}

The case of identically-labeled $G$ and $S$ produces the largest number of coalescent histories among all $p$-pseudocaterpillar gene trees and caterpillar species trees with fixed $(n,p)$ and $p \geq 4$. Note that for $p=3$, the case of identically-labeled $G$ and $S$ produces more coalescent histories than any non-identically-labeled pair; both $G$ and $S$ are caterpillars in this case, and fixing $n$, the number of coalescent histories for matching caterpillars exceeds the number of coalescent histories for any non-matching pair of caterpillars \citep[][Corollary 11]{himwich2020}.

We now return to the case of identically-labeled $(G,S)$ and evaluate eq.~\eqref{coalescent} for fixed small $p$.  

\subsection{Exact formulas for fixed $p$}
\label{secExactFormulas}

Fixing the variable $p$ allows us to obtain exact formulas for $h(n,p)$ as a rational function of $n$. The smallest case is $p=3$, so that the function $h(n,3)$ is defined for all $n\geq 3$:
\begin{equation*}
h(n,3) = \sum_{k=2}^{n-1} C(k,0) \, (k-1) \, C(n-3,n-k-1).
\end{equation*}
We rewrite the summand for $h$ in the expanded form from eq.~\eqref{coalescent_exp}. We then obtain the sum using the Wilf-Zeilberger algorithm for computing sums that involve binomial coefficients.
\begin{proposition}
\label{wz1}
For all $n \geq 3$, the following identity holds
\begin{equation}
\label{eq_h3}
h(n,3) = \sum_{k=2}^{n-1}\frac{(k-1)^2(2n-k-4)!}{(n-2)!(n-k -1)!} = \frac{3(2n-4)!}{ n! (n-3)!}.
\end{equation}
\end{proposition}
\begin{proof}
First, let $m = n-1$. Let the function $F(m,k)$ be the ratio of the summand to the right-hand side of eq.~\eqref{eq_h3}: 
\[
F(m,k) =  \frac{(k-1)^2 m (m+1) (m-2)! (2m-k-2)!}{6 (m-1)(2m-3)! (m-k)!}.
\]

This function and a proof certificate
\[
R(m,k) = -\frac{(k-2) (2m-k-1) \left(k^2 m-k^2+k-2 m\right)}{2 m^2 (2 m-1) (k-1) (m-k+1)}
\]
satisfy the assumptions of the Wilf-Zeilberger theorem \citep[Theorem~7.1.1]{aeqb}. Hence, the sum $\sum_{k=2}^m F(m,k)$ does not depend on $m$. We know that $\sum_{k=2}^m F(m,k) = 1$ when $m=2$, from which eq.~\eqref{eq_h3} follows by substituting $n = m+1$.
\end{proof}

It is convenient to write eq.~\eqref{eq_h3} as a product of a rational function of $n$ and a Catalan number, 
$$h(n,3)=\frac{3(n-2)}{2(2n-3)}\mathcal{C}_{n-1}.$$ 
For other small values of $p$, we follow the proof in Proposition~\ref{wz1} to obtain analogous expressions (Table~\ref{table:closed_forms}). The corresponding proof certificates $R(m,k)$ appear in Appendix~\ref{wz_table}.

\begin{table}[tb]
\centering
{\setlength{\extrarowheight}{10pt}%
\begin{tabular}{lll}
\toprule
 $p$ & $h(n,p)$ & $\lim_{n \rightarrow \infty} \frac{h(n,p)}{\mathcal{C}_{n-1}}$\\
 \midrule
3 & $\frac{3(n-2)}{2(2n-3)} \mathcal{C}_{n-1}$ & $\frac{3}{4}$ \\
 4 &  $\frac{(19 n-40)(n-3)}{4(2n-3)(2n-5)}\mathcal{C}_{n-1}$ & $\frac{19}{16}$ \\
 5 &  $\frac{(49 n^2-254 n+315)(n-4)}{4(2n-3)(2n-5)(2n-7)}\mathcal{C}_{n-1}$ & $\frac{49}{32}$ \\
 6 &  $\frac{(467 n^3-4319 n^2+12798 n-12096)(n-5)}{16(2n-3)(2n-5)(2n-7)(2n-9)}\mathcal{C}_{n-1}$ & $\frac{467}{256}$ \\
 7 &  $\frac{(1067 n^4-15263 n^3+78997 n^2-174673 n+138600)(n-6)}{16(2n-3)(2n-5)(2n-7)(2n-9)(2n-11)}\mathcal{C}_{n-1}$ & $\frac{1067}{512}$\\
 8 &  $\frac{(4751 n^5-96706 n^4+762163 n^3-2898044 n^2+5296836 n-3706560)(n-7)}{32(2n-3)(2n-5)(2n-7)(2n-9)(2n-11)(2n-13)}\mathcal{C}_{n-1}$ & $\frac{4751}{2048}$\\
 9 &  $\frac{(10393 n^6-284776 n^5+3155822 n^4-18055844 n^3+56078685 n^2-89321220 n+56756700)(n-8)}{32(2n-3)(2n-5)(2n-7)(2n-9)(2n-11)(2n-13)(2n-15)}\mathcal{C}_{n-1}$ & $\frac{10393}{4096}$\\
 \bottomrule
\end{tabular}}
\caption{Closed-form expressions for the function $h(n,p)$ for fixed values of $p$ (eq.~\eqref{coalescent}). Wilf-Zeilberger proof certificates appear in Appendix~\ref{wz_table}. The next three terms for $\lim_{n\rightarrow \infty} [h(n,p)/\mathcal{C}_{n-1}]$ are $179587/65536$ for $p=10$, $384199/131072$ for $p=11$, and $1631605/524288$ for $p=12$.}
\label{table:closed_forms}
\end{table}

\subsection{Asymptotic behavior for small $p$}
\label{sect:asymptotic}

We can extend beyond the exact formulas for $h(n,p)$ for small $p$ in Section \ref{secExactFormulas} to show that for each fixed $p$, there exists a constant $\beta_p$ such that $\lim_{n \rightarrow \infty} h(n,p) \sim \beta_p \mathcal{C}_{n-1}$. The approach follows \cite{disanto2016asymptotic}, who considered matching gene trees and species trees in \emph{caterpillar-like families}, in which trees had a caterpillar shape with the caterpillar subtree of size $\ell$ replaced by a ``seed tree'' $t$ of size $\ell$. They assumed $G = S = t^{(n)}$, with $t^{(n)}$ consisting of $t$ augmented by  $n$ ``caterpillar branches'' appended to its root.
    
The framework makes use of additional definitions. An \emph{$r$-extended coalescent history} is a coalescent history for the case in which a species tree is assumed to have its root-branch divided into $m \geq 1$ components \citep{rosenberg2007counting}. Labeling these components from 1 to $r$ with branch 1 closest to the species tree root, an \emph{$m$-rooted coalescent history} is an $r$-extended coalescent history in which the gene tree root coalesces on species tree branch $m$, $1 \leq m \leq r$. The number of $m$-rooted coalescent histories $h_{n,m}$ for $G=S=t^{(n)}$ then equals $h_{n,m} = e_{n,m}-e_{n,m-1}$, with $e_{n,0}=0$. 
 
\cite{disanto2016asymptotic} devised an iterative procedure for obtaining the coalescent histories for $t^{(n+1)}$ from the coalescent histories for $t^{(n)}$, $n \geq 0$. For a fixed seed tree $t$, the generating function for the sequence $h_{0,m}(t)$ counting $m$-rooted coalescent histories for $t$ is written
\begin{equation*}
        g(y) = \sum_{m=1}^\infty h_{0,m}(t) \, y^m.
\end{equation*}
The bivariate generating function for the sequence $h_{n,m}(t)$, counting  $m$-rooted histories for $(G,S)=(t^{(n)},t^{(n)})$, is denoted
\begin{equation*} 
        F(y,z) = \sum_{m = 1}^\infty \sum_{n=0}^\infty h_{n,m}(t) \, z^ny^m.
\end{equation*}
The univariate generating function $f(z)$ for the sequence $h_{n,1}(t)$ counts coalescent histories for $(G,S)=(t^{(n)},t^{(n)})$, satisfying 
\begin{equation*}
        f(z) = \sum_{n=0}^\infty h_{n,1}(t) \, z^n = \frac{\partial F(0,z)}{\partial y}.
\end{equation*}

\cite{disanto2016asymptotic} obtained the result 
\begin{equation}
    f(z) = \frac{g\left(\frac{1-\sqrt{1-4z}}{2}\right)}{z}.
\end{equation}
By examining the expansion of $f(z)$ around its dominant singularity, they showed that given $t$, there exists a positive constant $\beta_t$ such that    $h_{n,1}(t) \sim \beta_t \mathcal{C}_{n-1}$.

The construction of \cite{disanto2016asymptotic} that enumerated coalescent histories of $t^{(n+1)}$ from those of $t^{(n)}$ does not use $G=S$. Thus, it  applies for identically-labeled caterpillar-like families generated from \emph{nonmatching} seed trees $t_G$ and $t_S$ of the same size, as does the associated procedure for obtaining the generating function $f(z)$ for the number of coalescent histories for caterpillar-like families with $p$-pseudocaterpillar $G$ and identically-labeled caterpillar $S$.

We first derive an expression for  $e_{(n,p),r}$, the number of $r$-extended coalescent histories for the $p$-pseudocaterpillar gene tree of $n \geq p$ leaves on an identically-labeled caterpillar species tree of $n$ leaves. In our notation, dividing the root-branch amounts to adding right-steps on the diagram for $u_k$ (Section~\ref{secUpperPart}) and increasing the range of the index $k$. We have
\begin{equation}
    \label{eq:h_extended}
    e_{(n,p),r} = \sum_{k=p-1}^{n+r-2}C(k,p-3) \, (k-p+2) \, C_t(n-p,n+r-k-2,r).
\end{equation}
To obtain this expression, note that the extension of the species tree from 1 to $r$ branches ancestral to the root does not affect the $C(k,p-3)$ and $k-p+2$ terms, representing coalescences descended from the pivotal coalescence. However, the Catalan trapezoid that tabulates coalescent histories ancestral to the pivotal coalescence is affected. The number of coalescences ancestral to the pivotal coalescence continues to be $n-p$. The number of available branches is now $n+r-k-2$ instead of $n-k-1$. Traversing the paths ``forward,'' the trapezoid has order $k-p+2$, one more than the number of coalescences that can occur on the species tree branch on which the pivotal coalescence takes place, giving a count of  $C_t(n+r-k-2,n-p,k-p+2)$ (eq.~\eqref{trapezoid_entry}). The number of vertices on the upper edge of the trapzeoid is $r$, so that if paths are traversed in reverse order, the number of $r$-extended coalescent histories is, equivalently, $C_t(n-p,n+r-k-2,r)$.

The associated number of $m$-rooted coalescent histories, $1 \leq m \leq r$, then satisfies 
\begin{equation}
    \label{eq:h_rooted}
    h_{(n,p),m} = e_{(n,p),m} - e_{(n,p),m-1}.
\end{equation}
Noting that $n=p$ for the seed tree for the $p$-pseudocaterpillar family and  applying eqs.~\eqref{eq:h_extended} and eq.~\eqref{eq:h_rooted} gives generating function $g_p(y)$,
\begin{equation}
\label{eq:g_pcat}
    g_p(y) = \sum_{m=1}^\infty h_{(p,p),m} y^m = \sum_{m=1}^\infty \frac{m(m+2)(2p+m-5)!}{(p-3)!(p+m-1)!}y^m.
\end{equation}
For small $p$, the generating functions $g_p$ can be simplified as in Table~\ref{table:generating_funcs}. 

From $g_p(y)$, we then obtain the generating function $f_p(z)$ that counts coalescent histories as the numbers of leaves in the gene tree and species tree increase from $p$:
\begin{equation}
\label{eq:h_genfunc}
    f_p(z) = \sum_{k=0}^\infty h_{p+k,1}z^n = \frac{g_p\left(\frac{1-\sqrt{1-4z}}{2}\right)}{z}.
\end{equation}
Using eq.~\eqref{eq:g_pcat}, these generating functions can also be simplified for small $p$ (Table~\ref{table:generating_funcs}). 

Expanding the entries in Table~\ref{table:generating_funcs}, we obtain, for example:
\begin{eqnarray}
   \label{eq:f3_genfunc_example}
    f_3(z) & = & 1 + 3z + 9z^2 + 28z^3 + 90z^4 + 297z^5 + 1001z^6 + 
    3432z^7 + O(z^{8}) \nonumber \\
     \label{eq:f4_genfunc_example}
    f_4(z) & = & 3 + 11z + 37z^2 + 124z^3 + 420z^4 + 1441z^5 + 5005z^6 +
    17576z^7 + O(z^{8}). \nonumber
\end{eqnarray}
Each function gives the values $h(n,p)$ (eq.~\eqref{coalescent}) as $n$ is incremented beginning with $n=p$.

\begin{table}[tb]
    \centering
    \begin{tabular}{ccc}
    \toprule
        $p$ & $g_p(y)$ & $f_p(z)$ \\
    \midrule
3 & $\frac{y}{(y-1)^2}$ & $\frac{2(-\sqrt{1-4z}+1)}{z(\sqrt{1-4z}+1)^2}$ \\[2ex]
4 & $\frac{y^2-3y}{(y-1)^3}$ & $\frac{8(z-\sqrt{1-4z}+1)}{z(\sqrt{1-4z}+1)^3}$ \\[2ex]
5 & $\frac{2y^3-8y^2+9y}{(y-1)^4}$ & $\frac{8(\sqrt{1-4z}-1) (2z-3\sqrt{1-4z}-6)}{z(\sqrt{1-4z}+1)^4}$ \\[2ex]
6 & $\frac{5y^4-25y^3+44y^2-28y}{(y-1)^5}$ & $\frac{16(-10 z^2+33z+15z \sqrt{1-4z}-4 \sqrt{1-4z}+4)}{z(\sqrt{1-4z}+1)^5}$ \\[2ex]
    \bottomrule
    \end{tabular}
    \caption{Generating functions. Generating function $g_p(y)$ counts $m$-rooted histories $h_{(p,p),m}$ for caterpillar-like families with a seed $p$-pseudocaterpillar gene tree and identically-labeled caterpillar species tree (eq.~\eqref{eq:g_pcat}); generating function $f_p(z)$ counts coalescent histories $h_{(p,p),m}$ (eq.~\eqref{eq:h_genfunc}).}
    \label{table:generating_funcs}
\end{table}

\begin{table}[tb]
\centering
{\setlength{\extrarowheight}{3pt}%
\begin{tabular}{ccc}
\toprule
    & \multicolumn{2}{c}{$\beta_p$} \\\cline{2-3}
$p$ & Matching $p$-pseudocaterpillar $S$ & Caterpillar $S$\\
\midrule
3 & $1.0000$  & $0.7500$ \\
4 & $1.2500$  & $1.1875$ \\
5 & $1.4375$  & $1.5313$ \\
6 & $1.5938$  & $1.8242$ \\
7 & $1.7305$  & $2.0840$ \\
8 & $1.8535$  & $2.3198$ \\
9 & $1.9663$  & $2.5374$ \\
\bottomrule
\end{tabular}}
\caption{Numerical values of the constant $\beta_p$ describing $\lim_{n\rightarrow \infty} [h(n,p)/\mathcal{C}_{n-1}]$, the asymptotic ratio of the number of coalescent histories $h(n,p)$ to the Catalan number $\mathcal{C}_{n-1}$. The gene tree has a $p$-pseudocaterpillar topology. Values for the case that the gene tree and species tree $S$ have a matching $p$-pseudocaterpillar topology are taken from Table 1 of \cite{disanto2016asymptotic}; values for identically-labeled caterpillar $S$ are taken from Table \ref{table:closed_forms}.}
\label{table:comparison}
\end{table}

For $p$-pseudocaterpillar gene trees, we can compare the values for the limiting constants $\beta_p$ for two choices of the species tree $S$: the case in which the species tree has the same $p$-pseudocaterpillar labeled topology, and the case of an identically-labeled caterpillar species tree. The former value, from \cite{disanto2016asymptotic}, exceeds the latter for $p=3$ and $p=4$ (Table \ref{table:comparison}). For $p=5$ to $p=9$, however, $\beta_p$ for the non-matching caterpillar $S$ exceeds that for a matching $p$-pseudocaterpillar.

\cite{rosenberg2010coalescent} had shown that the case of a 4-pseudocaterpillar gene tree and an identically-labeled caterpillar species tree produced more coalescent histories ($\beta_p=1.1875$) than the case of matching caterpillar gene tree and species tree ($\mathcal{C}_{n-1}$ coalescent histories, and hence a limiting ratio of 1). The table demonstrates that $p$-pseudocaterpillar gene trees for each $p$ from 5 to 9 also produce more coalescent histories than the matching caterpillar gene tree.

\section{Maximal number of coalescent histories for fixed $n$}
\label{sec:max}

Applying Theorem \ref{thm:coalescent}, we can calculate $h(n,p)$ systematically for small $n$ and all $p$ with $3 \leq p \leq n$.  Table~\ref{table:values} shows the values of $h(n,p)$ for all $(n,p)$ with $n \leq 12$.

The table suggests two patterns. First, we can see that a symmetry exists in which $h(n,p)=h(n,n-p+3)$. We will verify this symmetry in Section \ref{secSymmetry}. Second, we can observe that for each $n$, the value of $p$ that maximizes $h(n,p)$ lies in the middle, repeating for two adjacent values of $p$ when $n$ is even. We state this result formally in the following theorem.

\begin{theorem}
	\label{thm:maxp}
	Consider a caterpillar species tree $S$ with $n \geq 4$ leaves. Among identically-labeled $p$-pseudocaterpillar gene trees $G$ with $n$ leaves and $3 \leq p \leq n$, the value of $p$ that maximizes the number of of coalescent histories $h(n,p)$ for $(G,S)$ is 
	\begin{equation}
	\label{maxp_odd}
		p_m = \frac{n+3}{2}
	\end{equation}
	if $n$ is odd. If $n$ is even, then two adjacent maxima exist:
	\begin{equation}
		\label{maxp_even}
		p_{m1} = \frac{n+2}{2}, \: p_{m2} = \frac{n+4}{2}.
	\end{equation}
\end{theorem}

For $n=3$ and $n=4$, the result is trivial, as $n=3$ requires $p=3$, and for $n=4$, $h(n,3)=h(n,4)=3$. For $n\geq 5$, the proof proceeds in three steps.
\begin{enumerate}
		\item First, in Section \ref{secDifferenceFunction}, for $n\geq 5$ and $4 \leq p \leq n$, we describe a \emph{difference function} $D(n,p)$ that measures the change in the function $h(n,p)$ when we increment $p$ by 1 for fixed $n$. 
		\item Next, in Section \ref{secSignDifferenceFunction}, we show that the difference function $D(n,p)$ is positive for $p=4$ (Lemma~\ref{thm_p_eq_4}) and negative for $p=n$ (Lemma~\ref{thm_p_eq_n}), and that it monotonically decreases as the integer $p$ is incremented from 4 to $n$ (Lemma~\ref{thm_monotonic}). 
		\item Finally, in Section \ref{secEvenValues}, we deduce that for fixed $n \geq 5$, if $D(n,p) \neq 0$ for all $p$, $4 \leq p \leq n$, then a unique integer $p$ exists at which $h(n,p)$ is maximal; two maxima exist if $D(n,p)=0$ for some $p$. We confirm that the maxima of $h(n,p)$ are described by eqs.~\eqref{maxp_odd} and \eqref{maxp_even}.
\end{enumerate} 

\begin{table}[tb]
\centering
\begin{tabular}{cllllllllllll}
\toprule
&\multicolumn{10}{c}{$p$}\\\cline{2-11}
$n$ & 3     & 4     & 5     & 6     & 7     & 8     & 9     & 10    & 11    & 12    \\
\midrule
 3  & 1     &       &       &       &       &       &       &       &       &       \\
 4  & 3     & 3     &       &       &       &       &       &       &       &       \\
 5  & 9     & 11    & 9     &       &       &       &       &       &       &       \\
 6  & 28    & 37    & 37    & 28    &       &       &       &       &       &       \\
 7  & 90    & 124   & 134   & 124   & 90    &       &       &       &       &       \\
 8  & 297   & 420   & 473   & 473   & 420   & 297   &       &       &       &       \\
 9  & 1001  & 1441  & 1665  & 1735  & 1665  & 1441  & 1001  &       &       &       \\
 10 & 3432  & 5005  & 5885  & 6291  & 6291  & 5885  & 5005  & 3432  &       &       \\
 11 & 11934 & 17576 & 20930 & 22766 & 23354 & 22766 & 20930 & 17576 & 11934 &       \\
 12 & 41990 & 62322 & 74932 & 82537 & 86149 & 86149 & 82537 & 74932 & 62322 & 41990 \\
\bottomrule
\end{tabular}
\caption{Values of the function $h(n,p)$ (eq.~\eqref{coalescent}) for small values of $n$ and $p$.}
\label{table:values}
\end{table}

\subsection{Difference function}
\label{secDifferenceFunction}

For $n \geq 5$ and $4 \leq p \leq n$, we define the difference function of $h(n,p)$:
\begin{equation}
    \label{diff}
    D(n,p) = h(n,p)-h(n,p-1).
\end{equation}
Because $h(n,p)$ is defined as a sum, $D(n,p)$ is also an expression involving a summation. 

We find a closed-form expression for $D(n,p)$. We start by expanding eq.~\eqref{diff} using eq.~\eqref{coalescent_exp}:
\begin{eqnarray}
    D(n,p) & = & \sum_{k=p-1}^{n-1} \Bigg[ \frac{(k-p+2)^2 (k-p+4) (k+p-3)! ( 2n-p-k-1)!}{(k+1)! (p-3)! (n-k-1)! (n-p+1)!} \nonumber \\
    & & - \frac{(k-p+3)^2 (k-p+5) (k+p-4)! (2 n-p-k)!}{(k+1)! (p-4)! (n-k-1)! (n-p+2)!}\Bigg]\nonumber \\ & & -  \frac{3(2n-2p+2)!(2p-6)!}{(n-p+1)!(n-p+2)!(p-4)!(p-1)!}.\nonumber 
\end{eqnarray}
Notice that because the sums in expressions for $h(n,p)$ and $h(n,p-1)$ have different summation limits, we obtain an additional term outside the sum. The sum in the expression for $D(n,p)$ can be transformed into a closed form, which gives the following formula:
\begin{multline}
\label{diff_ex}
D(n,p) = \frac{2p(np+7p-2p^2-6)(2n-2p+2)! \, (2p-5)!}{n(n-1)(n-p)!  \, (n-p+2)! \, (p-3)! \, p!} - \frac{3(2n-2p+2)!(2p-6)!}{(n-p+1)!(n-p+2)!(p-4)!(p-1)!}.
\end{multline}
The proof appears in Appendix~\ref{wzdiffproof}.

\subsection{Sign of the difference function}
\label{secSignDifferenceFunction}

Using eq.~\eqref{diff_ex} for $D(n,p)$, we prove three lemmas concerning the sign of $D(n,p)$.

\begin{lemma}
\label{thm_p_eq_4}
For $p=4$, the function $D(n,p)$ is positive for all $n\geq5$.
\end{lemma}
\begin{proof}
When we substitute $p=4$ into eq.~\eqref{diff_ex}, we obtain 
\[
D(n,4) = \frac{4(2n-5)!}{n!(n-4)!} - \frac{(2n-6)!}{(n-2)!(n-3)!} =\frac{2(7n-15)(2n-7)!}{n!(n-5)!}.
\]
Because $n\geq5$, all terms in this fraction are positive.
\end{proof}

\begin{lemma}
\label{thm_p_eq_n}
For $p=n$, the function $D(n,p)$ is negative for all $n\geq 5$.
\end{lemma}
\begin{proof}
Substituting $p=n$ into eq.~\eqref{diff_ex}, we obtain
\[
D(n,n) = -\frac{2(n-6)(2n-5)!}{n!(n-3)!} - \frac{3(2n-6)!}{(n-4)!(n-1)!}.
\]
For $n\geq 7$, $D(n,n)$ is quickly seen to be a sum of two negative numbers. It remains to check the cases of $n=5$ and $n=6$: $D(5,5) = -2$ and $D(6,6)=-9$.
\end{proof}

\begin{lemma}
\label{thm_monotonic} $D'(n,p) = D(n,p)-D(n,p-1)$ is  negative for $n\geq 5$ and $5 \leq p \leq n$. That is, for each $n \geq 5$, $D(n,p)$ monotonically decreases as the integer $p$ is incremented from $p=4$ to $p=n$.
\end{lemma}
\begin{proof}
The expression for $D'(n,p)$ can be simplified to
\begin{equation*}
    D'(n,p) = \frac{2 (4p^2-4np-20p+11n+27) (2n-2p+2)!(2p-8)!}{(n-p+1)!(n-p+3)!(p-4)!(p-2)!}.
\end{equation*}
Because $5 \leq p \leq n$, the term that determines the sign of $D'(n,p)$ is the polynomial $f(n,p) = 4p^2 - 4np - 20p  + 11n + 27$ in the numerator. Solving the inequality $f(n,p) < 0$ for $p$, we obtain 
\[
\frac{n+5}{2} - \frac{1}{2}\sqrt{n^2-n-2} < p <  \frac{n+5}{2} + \frac{1}{2}\sqrt{n^2-n-2}.
\]
The left-hand term is bounded above by 3 for all $n \geq 5$, and the right-hand term exceeds $n$ for all $n\geq 5$. Hence, because $5 \leq p \leq n$, all possible values of $(n,p)$ satisfy the inequality.
\end{proof}

\subsection{Location of the maximum}
\label{secEvenValues}

As a result of Lemmas \ref{thm_p_eq_4}-\ref{thm_monotonic}, for $n\geq 5$, as $p$ is incremented from 4 to $n$, $D(n,p)$ monotonically decreases (Lemma \ref{thm_monotonic}) from a positive value at $p=4$ (Lemma \ref{thm_p_eq_4}) to a negative value at $p=n$ (Lemma \ref{thm_p_eq_n}). Hence, $h(n,p)$ increases from $p=3$ to a maximum then decreases until $p=n$. 

Two cases are possible. Given $n$, a unique value $p=p_{m1}$ could exist at which $D(n,p)=0$, in which case $h(n,p_{m1})=h(n,p_{m1}-1)$, and both $p_{m1}$ and $p_{m2}=p_{m1}-1$ are maxima. Alternatively, if $D(n,p) \neq 0$ for all $p$, then $h(n,p)$ is maximized at the largest value of $p$ for which $D(n,p)>0$.

If $n \geq 6$ is even, then inserting $p=\frac{n+4}{2}$ into eq.~\eqref{diff_ex}, we obtain $D(n,\frac{n+4}{2})=0$. Hence, $h(n,\frac{n+4}{2}) = h(n,\frac{n+2}{2})$, and maxima of $h(n,p)$ occur at both $p_{m1}=\frac{n+4}{2}$ and $p_{m2}=\frac{n+2}{2}$.

If $n \geq 5$ is odd, we show that a value $p_m \geq 4$ exists for which $D(n,p_m)>0$ and $D(n,p_m+1)<0$. This value $p_m$ maximizes $D(n,p)$.

\begin{lemma}
\label{thm_odd_n}
For odd $n \geq 5$, $n=2k+1$ and $k\geq 2$, (i) $D(n,k+2) > 0$, and (ii) $D(n,k+3)<0$.
\end{lemma}
\begin{proof}
(i) We insert $(n,p)=(2k+1,k+2)$ into eq.~\eqref{diff_ex}, obtaining the positive quantity
\begin{equation*}
D(2k+1,k+2)=\frac{(2k)!\,(2k-2)!}{(k-1)! \, (k!)^2 \, (k+1)!}.
\end{equation*}
(ii) Inserting $(n,p)=(2k+1,k+3)$ into eq.~\eqref{diff_ex}, we obtain
\begin{equation*}
    D(2k+1,k+3) = -\frac{2\, (2k)! \, (2k-3)!}{(k-2)! \, (k!)^2 \, (k+1)!}.
\end{equation*}
a quantity that is negative.
\end{proof}

We conclude $h(2k+1,k+2) > h(2k+1,k+1)$, but $h(2k+1,k+3) < h(2k+1,k+2)$. Hence, for odd $n \geq 5$,  writing $k=\frac{n-1}{2}$, $p_m=\frac{n+3}{2}$ maximizes $h(n,p)$. The proof of Theorem \ref{thm:maxp} is complete.

\subsection{Asymptotic growth of the maximal number of coalescent histories}
\label{secMaximal}

With the value $p_m$ that maximizes $h(n,p)$ established, we now examine the asymptotic growth of the maximum. We quickly verify that for a fixed caterpillar species tree with $n$ leaves, across all $p$-pseudocaterpillar gene trees with fixed $n$, the maximal number of coalescent histories grows faster than the Catalan number $\mathcal{C}_{n-1}$ describing the number of coalescent histories for the matching caterpillar. In Section~\ref{sect:asymptotic}, we showed that for fixed small $p \geq 4$, as $n$ increases, the number of coalescent histories grows with a constant multiple of $\mathcal{C}_{n-1}$, with the constant exceeding 1. Here we show that for each $n \geq 7$, the maximal number of coalescent histories, that is, $h(n,p_{m})$ for odd $n$ and $h(n,p_{m1})=h(n,p_{m2})$ for even $n$, exceeds the corresponding Catalan number. 

\begin{proposition}
\label{thm:h_geq_catalan}
	For odd $n \geq 7$, $h(n,p_{m}) > \mathcal{C}_{n-1}$, and for even $n \geq 8$, $h(n,p_{m1}) = h(n,p_{m2}) > \mathcal{C}_{n-1}$, where $h$ is defined by eq.~\eqref{coalescent}, $p_{m}$ by eq.~\eqref{maxp_odd}, $p_{m1}$ and $p_{m2}$ by eq.~\eqref{maxp_even}, and $\mathcal{C}_{n-1}$ by eq.~\eqref{catalan}.
\end{proposition}

\begin{proof}
For $n=7$, we have $h(n,p_m)=h(7,5)=134$, which exceeds $C_6=132$. For $n=8$, $h(n,p_{m1})=h(n,p_{m2})=h(8,5)=h(8,6)=473$, which exceeds $C_7=429$.

Lemma 4.2 of \cite{rosenberg2010coalescent} showed that for $n \geq 9$, $h(n,4) > \mathcal{C}_{n-1}$. By Theorem \ref{thm:maxp}, for odd $n \geq 9$, $h(n,p_m) > h(n,4)$, and for even $n \geq 10$, $h(n,p_{m1}) = h(n,p_{m2}) > h(n,4)$. Thus, because the maximal number of coalescent histories across all $p$ exceeds the number for $p=4$, and because the number of coalescent histories for $p=4$ exceeds $\mathcal{C}_{n-1}$, the maximum exceeds $\mathcal{C}_{n-1}$.
\end{proof}

For short, we abbreviate $p_m=p_{m1}=p_{m2}$ for even $n$, so that the sequence of values of $h(n,p_m)$ is well-defined for $n \geq 4$. We introduce the definition of the exponential order of the sequence: a sequence $\{a_n\}$ has exponential order $k$ if
$\limsup_{n\to \infty} \, \sqrt[n]{a_n} = k$ \citep{flajolet_sedgewick}. In other words, $a_n = k^ns(n)$ where $s(n)$ is a subexponential factor with $\limsup_{n\to \infty} \sqrt[n]{s(n)} = 1$. If sequences $a_n$ and $b_n$ have the same exponential order, we write $a_n \bowtie b_n$.

The Catalan numbers $\mathcal{C}_n$ have exponential order $4$, as Stirling's approximation to $\mathcal{C}_n = (2n)!/[(n+1)(n!)^2]$ gives $\mathcal{C}_n \approx 4^n /(n^{3/2} \sqrt \pi)$. Plotting $\log h(n,p_m)$ and $\log \mathcal{C}_{n-1}$ as functions of $n$, we see that they grow approximately linearly with similar slopes (Figure~\ref{fig:hmax_plot}). We therefore claim that the sequence $h(n,p_m)$ also has exponential order $4$.
\begin{proposition}
\label{prop:hmax_order}
For caterpillar species trees with $n$ leaves, the sequence $h(n,p_m)$ describing the maximal number of coalescent histories across all $p$-pseudocaterpillar gene trees of size $n$ has exponential order 4, so that $h(n,p_m) \bowtie \mathcal{C}_n$.
\end{proposition}
\begin{proof}
By Lemmas \ref{lemma:h_lower_bound} and \ref{lemma:h_upper_bound} in Appendix C, $\mathcal{C}_{n-2} \leq h(n,p_m) \leq n\,\mathcal{C}_{n+2}$ for $n \geq 3$, from which
\begin{equation*}
    \lim_{n\to\infty} \sqrt[n]{\mathcal{C}_{n-2}} \leq \lim_{n\to\infty} \sqrt[n]{h(n,p_m)} \leq \lim_{n\to\infty} \sqrt[n]{n\,\mathcal{C}_{n+2}}.
\end{equation*}
As the left-hand and right-hand limits both equal 4, we conclude $\lim_{n\to\infty} \sqrt[n]{h(n,p_m)} = 4$, $h(n,p_m) = 4^n s(n)$ for some subexponential $s(n)$, and $h(n,p_m) \bowtie \mathcal{C}_n$.
\end{proof}

\begin{figure}[tb]\centering
\includegraphics[width=0.5\textwidth]{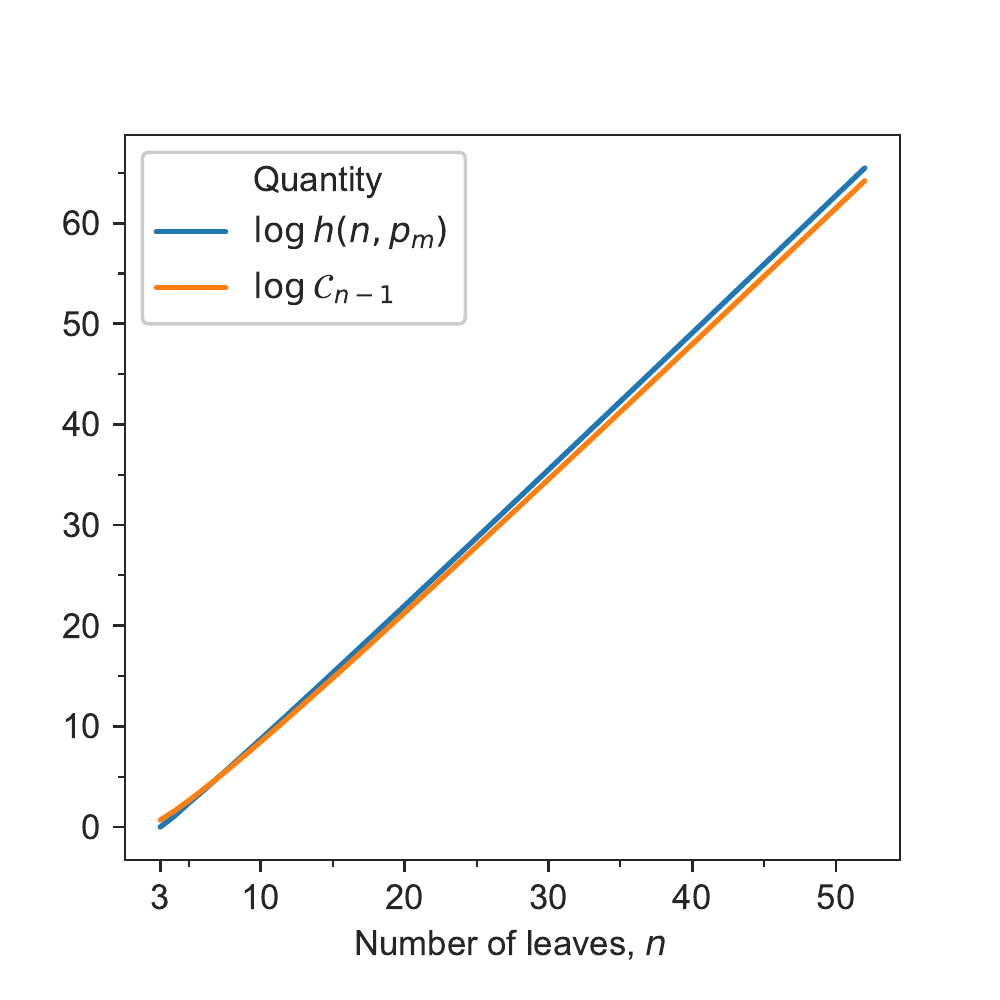}
\vskip -.5cm
\caption{For caterpillar species trees of size $n$, the natural logarithms of the maximal number of coalescent histories across  $p$-pseudocaterpillar gene trees ($h(n,p_m)$) and the number of coalescent histories of the matching caterpillar gene tree topology ($\mathcal{C}_{n-1}$). The quantity $h(n,p_m)$ is computed according to Theorem \ref{thm:maxp}, and $\mathcal{C}_{n-1}$ follows eq.~\eqref{catalan}.}
\label{fig:hmax_plot}
\end{figure}

\section{Symmetry}
\label{secSymmetry}

We now verify the symmetry $h(n,p)=h(n,n-p+3)$ observed in Table~\ref{table:values} for all $(n,p)$ with $3 \leq p \leq n$. For convenience, given a pseudocaterpillar tree with second cherry at position $p$, we define its \emph{dual} as the pseudocaterpillar tree with second cherry at position $n-p+3$ (Figure~\ref{fig:dual}). 

We show that for a fixed caterpillar species tree, the number of coalescent histories of an identically-labeled pseudocaterpillar is equal to the number of coalescent histories of its dual. The formula for the number of coalescent histories has a symmetry in the position of the second cherry on the $p$-pseudocaterpillar gene tree.
\begin{theorem}
	\label{symmetry}
	For all $(n,p)$ with $3\leq p \leq n$, $h(n,p) = h(n,n-p+3)$.
\end{theorem}

For $n=3$, the claim is trivial, as $p=n-p+3=3$. For $n=4$, the claim is also trivial, as $h(4,3)=h(4,4)=3$. For $n \geq 5$, we proceed in three steps.
\begin{enumerate}
	\item First, in Section~\ref{secDualDifference}, we introduce a \emph{dual difference} function $D^*(n,p)$ that measures the change in $h(n, n-p+3)$ as $p$ is incremented for fixed $n$. 
	\item Next, in Section~\ref{secDualEquality}, we show that the dual difference function $D^*(n,p)$ is equal to the regular difference function $D(n,p)$ for all allowed values of $(n,p)$.
	\item Finally, in Section~\ref{secBoundary}, we use this equality of difference functions to complete the proof of the symmetry of $h(n,p)$.
\end{enumerate}

\begin{figure}[tb]\centering

\begin{minipage}[b]{0.4\linewidth} (A)
	
	{\centering \includegraphics[width=\linewidth]{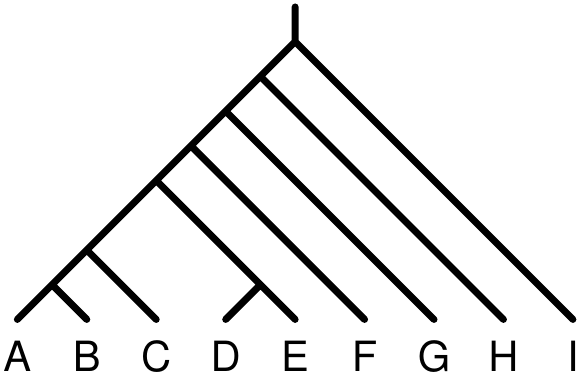}}
\end{minipage}
\hspace{1cm}
\begin{minipage}[b]{0.4\linewidth} (B)

	{\centering \includegraphics[width=\linewidth]{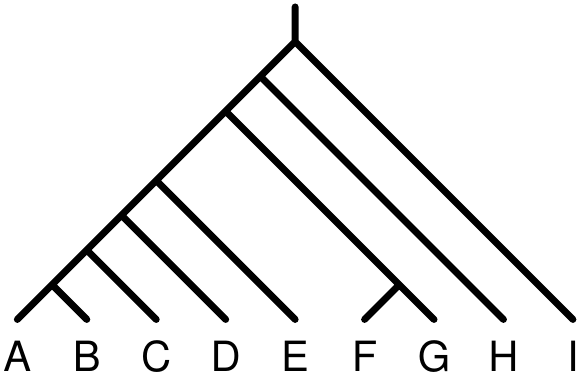}}
\end{minipage}
\caption{Dual $p$-pseudocaterpillar gene trees with the same number of coalescent histories (1665) if paired with the identically-labeled caterpillar species tree. (A)  $(n,p)=(9,5)$. (B) $(n,p)=(9,7)$.}
\label{fig:dual}
\end{figure}

\subsection{Dual difference function}
\label{secDualDifference}

We define a function $D^*$ ``dual'' to the difference function $D$ (eq.~\eqref{diff}):
\begin{equation}
\label{dualdiff}
D^*(n,p) = h(n,n-p+3) - h(n,n-p+4).
\end{equation}
The function is well defined for $n\geq5$ and $4 \leq p\leq n$, where $n-p+3$ and $n-p+4$ lie in $[3,n]$. 

Using the definition of $h(n,p)$ from eq.~\eqref{coalescent_exp}, we obtain
\begin{multline*}
	D^*(n,p) = \sum_{k=n-p+3}^{n-1} 
	\Bigg[ \frac{(k-n+p-1)^2 (k-n+p+1) (k+n-p)! (n+p-k-4)!}{(k+1)! (p-2)! (n-k-1)! (n-p)!}\\
	- \frac{(k-n+p-2)^2 (k-n+p) (k+n-p+1)! (n+p-k-5)!}{(k+1)! (p-3)! (n-k-1)! (n-p+1)!}\Bigg] \\
	+ \frac{3(2 p-6)!(2 n-2 p+2)!}{(p-3)!(p-2)!(n-p)!(n-p+3)!}.
\end{multline*}
This sum can be simplified to get a closed form for $D^*$:
\begin{multline}
\label{dualdiff_ex}
D^*(n,p) = \frac{4 (p-3)(n^2+2p^2-3np+5n-9p+10) (2 p-7)!(2n-2p+3)!}{n(n-1) (n-p+1)(p-4)!(p-2)!(n-p)!(n-p+3)!} \\
+ \frac{3(2 p-6)!(2 n-2 p+2)!}{(p-3)!(p-2)!(n-p)!(n-p+3)!}.
\end{multline}
The proof appears in Appendix~\ref{dualdiffproof}.

\subsection{Dual difference function is equal to the difference function}
\label{secDualEquality}

This section verifies the equality of the difference function and its dual.
\smallskip
\begin{lemma}
\label{thm_d_d_star}
For all $n\geq 5$ and $4\leq p \leq n$, the dual difference function equals the difference function 
\[
	D^*(n,p) = D(n,p).
\]
\end{lemma}
\begin{proof}
We simplify $D(n,p)/D^*(n,p)$ using eqs.~\eqref{diff_ex} and \eqref{dualdiff_ex}, verifying that this ratio equals 1.
\end{proof}

\subsection{Completing the proof}
\label{secBoundary}

Rearranging terms in the definitions of the difference functions by eqs.~\eqref{diff} and \eqref{dualdiff}, we have
\begin{equation}
h(n,p) - h(n,n-p+3) = h(n,p-1) - h(n,n-p+4) \label{diff_coal}
\end{equation}
Decrementing $p$ from n to 4, eq.~\eqref{diff_coal} gives a chain of equalities $h(n,n)-h(n,3)=h(n,n-1)-h(n,4)=h(n,n-2)-h(n,5)=\ldots=h(n,3)-h(n,n)$. 

In particular, for each $p$ from 3 to $n$, $h(n,p)-h(n,n-p+3)=-[h(n,p)-h(n,n-p+3)]$. Both sides of this equation must then equal zero, from which we conclude $h(n,p)=h(n-p+3)$ for each $p$, $3 \leq p \leq n$. The proof of Theorem~\ref{symmetry} is complete.

Theorem~\ref{symmetry} can strengthen Proposition~\ref{thm:h_geq_catalan}. We now know that $h(n,p) > \mathcal{C}_{n-1}$ for all $(n,p)$ with $n\geq 9$ and $4 \leq p \leq n-1$. From \cite{rosenberg2010coalescent}, $h(n,4) > \mathcal{C}_{n-1}$ for $n \geq 9$. By Theorem~\ref{symmetry}, $h(n,n-1)=h(n,4) > \mathcal{C}_{n-1}$. In the proof of Theorem \ref{thm:maxp}, we show that for $n$ odd, $h(n,p)$ increases as $p$ increases from 4 to $\frac{n+3}{2}$, and Theorem~\ref{symmetry} indicates that $h(n,p)$ decreases as $p$ increases from $\frac{n+3}{2}$ to $n-1$; similarly, for $n$ even, $h(n,p)$ increases as $p$ increases from 4 to $\frac{n+2}{2}$, with $h(n,\frac{n+2}{2})=h(n,\frac{n+4}{2})$, then decreases as $p$ increases from $\frac{n+4}{2}$ to $n-1$. Thus, $h(n,p) > \mathcal{C}_{n-1}$ for all $(n,p)$ with $n \geq 9$ and $4 \leq p \leq n-1$. 

\section{Discussion}
\label{secDiscussion}

We have developed a method for counting coalescent histories in cases in which the gene tree and species tree topologies do not match, considering  $p$-pseudocaterpillar gene trees together with an identically-labeled caterpillar species tree. Using a combinatorial construction, we find that the recursive formula from \cite{rosenberg2007counting} can be evaluated non-recursively as a sum  (eq.~\eqref{coalescent_exp})---which can in turn be simplified to a closed form for fixed small $p$ (Section \ref{secSpecialCases}). The number of coalescent histories $h(n,p)$ (eq.~\eqref{coalescent}) has a symmetry in $p$ (Theorem~\ref{symmetry}), and the maximum over values of $p$ for each $n$ is attained when the ``second cherry'' lies in the ``middle'' of the gene tree (Theorem \ref{thm:maxp}).

Results on the value of $p$ that maximizes $h(n,p)$ verify an informal observation from previous studies. It has been noted that for fixed $n$, large numbers of coalescent histories tend to occur when two conditions are met: the number of distinct sequences in which coalescences can take place is large, as is the number of species tree branches describing potential placements of those coalescences \citep{rosenberg2007counting, rosenberg2013caterpillar, rosenberg2010coalescent, disanto2015coalescent, disanto2016asymptotic}. For a fixed caterpillar species tree, identically-labeled $p$-pseudocaterpillar gene trees represent a tradeoff of these two features. As $p$ increases, more sequences exist for coalescences descended from the pivotal coalescence. However, the number of species tree branches on which the pivotal coalescence can occur decreases, so that fewer species tree branches exist on which the larger number of coalescence sequences can occur. That $h(n,p)$ is maximized when $p$ lies in the ``middle'' aligns with the informal observation that both conditions---many coalescence sequences, and many species tree branches on which coalescences take place---are important for generating large numbers of coalescent histories. 

\begin{table}[tb]
\centering
\begin{tabular}{ccccccccc}
\toprule
&\multicolumn{7}{c}{$p$}& Matching \\
\cline{2-8}
$n$ & 3 & 4 & 5 & 6 & 7 & 8 & 9 & caterpillar \\
\midrule
 \vspace{0.03in} 
 3  &  \doubleentry{1}{2}  &       &       &       &       &       & 		& 2   \\ 
 \vspace{0.03in} 
 4  & \doubleentry{3}{5}   & \doubleentry{3}{4}    &       &       &       &       & 		& 5   \\  \vspace{0.03in} 
 5  & \doubleentry{9}{14}  & \doubleentry{11}{13}  & \doubleentry{9}{10}      &       &       &       & 		& 14  \\  \vspace{0.03in} 
 6  & \doubleentry{28}{42}    & \doubleentry{37}{42}    & \doubleentry{37}{37}    & \doubleentry{28}{28}     &       &       & 		& 42  \\ \vspace{0.03in} 
 7  & \doubleentry{90}{132}     & \doubleentry{124}{138}    & \doubleentry{134}{130}   & \doubleentry{124}{112}    & \doubleentry{90}{84}     &       &		& 132 \\ \vspace{0.03in} 
 8  & \doubleentry{297}{429}    & \doubleentry{420}{461}    & \doubleentry{473}{453}    & \doubleentry{473}{416}    & \doubleentry{420}{354}    & \doubleentry{267}{264}    &       & 429 \\ \vspace{0.03in} 
 9  & \doubleentry{1001}{1430}  & \doubleentry{1441}{1573}  & \doubleentry{1665}{1584}  & \doubleentry{1735}{1511}  & \doubleentry{1665}{1368}  & \doubleentry{1441}{1155}  & \doubleentry{1101}{858}  & 1430\\ 
 \bottomrule
\end{tabular}
\caption{Numbers of coalescent histories for matching caterpillar gene trees and species trees ($\mathcal{C}_{n-1}$, right-hand column), caterpillar species trees and identically-labeled $p$-pseudocaterpillar gene trees (top entry in each cell), and matching $p$-pseudocaterpillar gene trees and species trees (bottom entry). Top entries are from Table~\ref{table:values}, and bottom entries are from \cite{rosenberg2007counting}.}
\label{table:topologies_comparison}
\end{table}

Table \ref{table:topologies_comparison} compares coalescent histories in three cases: matching caterpillars, matching $p$-pseudocaterpillars, and caterpillar species trees with identically-labeled non-matching $p$-pseudocaterpillar gene trees. For a caterpillar species tree, as the number of species $n$ grows to 9 or greater, the number of coalescent histories for identically-labeled $p$-pseudocaterpillar gene trees with $4 \leq p \leq n-1$ exceeds the Catalan number of coalescent histories for the matching gene tree (Proposition \ref{thm:h_geq_catalan}). For fixed $p$, more coalescent histories can occur for the non-matching $p$-pseudocaterpillar gene tree and identically-labeled caterpillar species tree than for matching $p$-pseudocaterpillars (Tables \ref{table:comparison} and \ref{table:topologies_comparison}).


\begin{figure}[tbp]\centering
\hspace{0.3cm}\begin{minipage}[b]{0.4\linewidth} (A) 

{\centering \includegraphics[width=\linewidth]{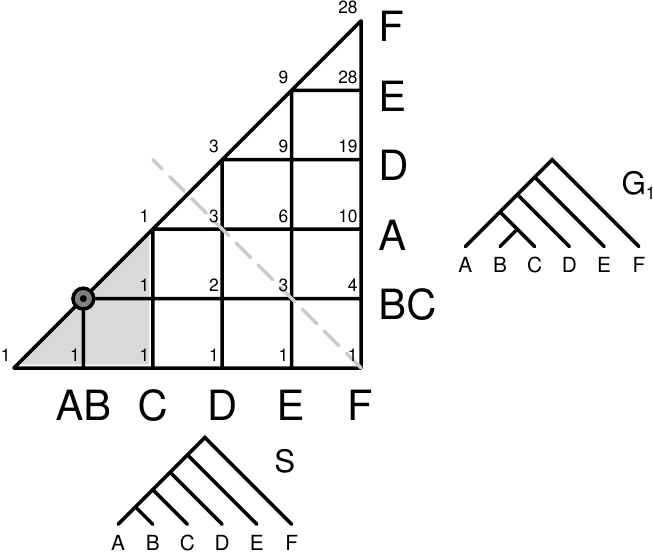}\label{fig:symmetry_ex_1}}
\end{minipage}\hspace{1cm}
\begin{minipage}[b]{0.4\linewidth} (B)

{\centering \includegraphics[width=\linewidth]{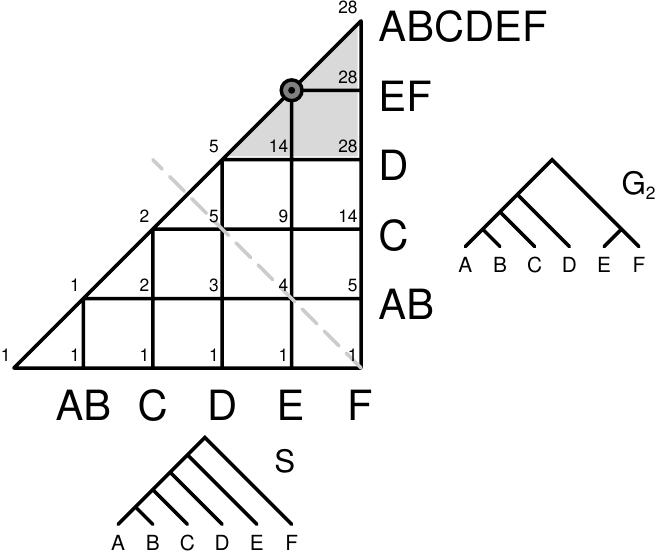}\label{fig:symmetry_ex_2}}
\end{minipage}

\caption{Symmetry in the number of coalescent histories for a caterpillar species tree and identically-labeled $p$-pseudocaterpillar gene trees, for the cases of $(n,3)$ and $(n,n)$. (A) For $p=3$, coalescent histories correspond to roadblocked monotonic paths on a lattice with one roadblock. The first coalescence $BC$ on a gene tree cannot happen on the species tree branch ancestral to $A$ and $B$. (B) The $p=n$ case can also be seen to correspond to roadblocked monotonic paths. The second-to-last coalescence $EF$ can occur only on the branch ancestral to the species tree root.}
\label{fig:symmetry_ex}
\end{figure}

In related work, \cite{disanto2019local} considered matching caterpillar gene trees and species trees, identifying the leaf whose replacement by a cherry in both trees would give rise to the greatest increase in the number of coalescent histories (measured as a ratio). This speciation---the splitting of a leaf node of $G$ and $S$ into two child nodes---can be interpreted as extending the trees by adding the ``second cherry'' that converts a caterpillar into a $p$-pseudocaterpillar. \cite{disanto2019local} determined the value of $p$ with which the $p$-pseudocaterpillar tree pair with $n$ leaves would have the largest number of coalescent histories. Asymptotically, this value of $p$ is equal to $\frac{n}{2}$ \citep[$i^*(n)$ for $\mathcal{D}_n$ in Table 1]{disanto2019local}---``in the middle,'' as in our result in Section~\ref{sec:max}. Theorem~\ref{thm:maxp} can then be seen to prove an analogous result in a nonmatching case, as the gene tree gains a cherry node whereas the species tree gains only a caterpillar leaf.

Our $p=3$ case has a direct geometric interpretation in the framework of \cite{himwich2020}, as it describes a non-matching pair of caterpillar trees (Figure~\ref{fig:symmetry_ex}). Its coalescent histories are described by monotonic paths on a lattice with a single roadblock. The $p=n$ case, for which the number of coalescent histories is equal to the $p=3$ case, can also be represented in a diagram with one roadblock, obtained by reflecting monotonic paths of the $p=3$ case across $y=n-1-x$. 
Because monotonic paths not crossing the $y=x$ diagonal of a square lattice correspond to Dyck paths, and the coalescent histories for $p=3$ correspond to Dyck paths beginning with two up-steps, the sequence $a_n = h(n,3) = h(n,n)$ for $n \geq 3$ gives the number of Dyck paths of length $n-1$ beginning with two up-steps (OEIS A000245). We can also write $h(n,3) = h(n,n) = a(n) = \mathcal{C}_{n-1} - \mathcal{C}_{n-2}$.

Our work provides an extension of an earlier study of coalescent histories for non-matching caterpillars \citep{himwich2020}. We expect that the method we have used has potential for extension to cases with more than two cherries, with the species tree remaining a caterpillar. In such an extension, each additional cherry would  generate an additional ``pivotal'' coalescence and an additional summation based on the placement of that coalescence.

\vskip .4cm
\noindent {\bf Acknowledgements.}
We acknowledge NIH grant R01 GM131404 for support.
{\footnotesize
\bibliographystyle{tpb}
\bibliography{coalescent_histories}
}
\clearpage
\appendix
\section{Wilf-Zeilberger certificates for formulas in Table~\ref{table:closed_forms}}
\label{wz_table}

This appendix gives the proof certificates for the identities in Table~\ref{table:closed_forms}, all of which have similar proofs to Proposition~\ref{wz1}. Only the Wilf-Zeilberger proof certificate $R(m,k)$ differs across the cases. We list the proof certificates for the remaining identities in Table \ref{table:appendix}.

\renewcommand*{\arraystretch}{1.4}
\begin{longtable}{ r  p{6in} }
  \caption[]{Wilf-Zeilberger proof certificates $R(m,k)$ for expressions in Table~\ref{table:closed_forms}}  \label{table:appendix}
    \\\hline 
   $p$ & $R(m,k)$\\\hline
    4     & $ -\{(k-3) (k-2 m+2) [k^3 (19 m^2-59 m+42)+k^2 (-19 m^2+86 m-72)+k (-65 m^2+134 m-68)-13 m^2+m+14]\}/[2 (k-2) k (m-1)^2 (2 m-3) (19 m-2) (k-m-1)]$  \\
    5     &  $-\{(k-4) (k-2 m+3) [k^4 (49 m^3-303 m^2+578 m-330)+k^3 (-49 m^3+388 m^2-858 m+540)+k^2 (-428 m^3+2483 m^2-4470 m+2394)-2 k (59 m^3+13 m^2-459 m+342)+4 (86 m^3-477 m^2+847 m-480)]\}/[2 (k-3) (k-1) (k+2) (m-2)^2 (2 m-5) (49 m^2-58 m+3) (k-m-1)]$\\
    6     &  $ -\{(k-5) (k-2 m+4) [k^5 (467 m^4-4786 m^3+17233 m^2-25394 m+12600)+ 7 k^4 (133 m^3-867 m^2+1754 m-1080)+k^3 (-8870 m^4+90676 m^3-324682 m^2+473420 m-230304)+k^2 (-9780 m^4+80459 m^3-234273 m^2+285790 m-127176)+3 k (7531 m^4-84096 m^3+314435 m^2-463542 m+227352)+18 (1705 m^4-14774 m^3+47147 m^2-65566 m+32088)]\}/[2 (k-4) (k-2) (k+2) (k+3) (m-3)^2 (2 m-7) (467 m^3-1517 m^2+1126 m-40) (k-m-1)]$\\
    7     &  $-\{(k-6) (k-2 m+5) [k^6 (1067 m^5-16330 m^4+94585 m^3-256250 m^2+319728 m-143640)+k^5 (2134 m^5-30293 m^4+163673 m^3-415318 m^2+486024 m-204120)+k^4 (-34377 m^5+532064 m^4-3106039 m^3+8451124 m^2-10543892 m+4705320)+k^3 (-104394 m^5+1517539 m^4-8380799 m^3+21722234 m^2-26044120 m+11362200)+2 k^2 (62807 m^5-1114568 m^4+7105518 m^3-20362003 m^2+25974226 m-11663880)+4 k (165353 m^5-2490002 m^4+14086017 m^3-37154572 m^2+45368164 m-20112720)+288 (2001 m^5-25450 m^4+128585 m^3-322100 m^2+388324 m-173040)]\}/[2 (k-5) (k-3) (k+2) (k+3) (k+4) (m-4)^2 (2 m-9) (1067 m^4-6727 m^3+13027 m^2-7697 m+210) (k-m-1)]$\\
    8     &  $-\{(k-7) (k-2 m+6) [k^7 (4751 m^6-101457 m^5+860555 m^4-3681375 m^3+8289854 m^2-9182568 m+3825360)+k^6 (23755 m^6-495790 m^5+4118327 m^4-17279672 m^3+38197692 m^2-41536512 m+16964640)+k^5 (-206286 m^6+4477837 m^5-38491018 m^4+166405493 m^3-377604962 m^2+420063216 m-174878640)-3 k^4 (468054 m^6-9846179 m^5+82298858 m^4-346924603 m^3+769533814 m^2-839428224 m+344802960)+k^3 (-228007 m^6+517454 m^5+26651409 m^4-221650784 m^3+680141896 m^2-877046808 m+386542080)+k^2 (13435271 m^6-292844363 m^5+2510774503 m^4-10762507513 m^3+24114363462 m^2-26483935680 m+10933917840)+k (30799166 m^6-626694930 m^5+5077031510 m^4-20873258910 m^3+45641745164 m^2-49626013680 m+20475669600)+600 (37811 m^6-656541 m^5+4768367 m^4-18525795 m^3+39673190 m^2-43034952 m+17843760)]\}/[2 (k-6) (k-4) (k+2) (k+3) (k+4) (k+5) (m-5)^2 (2 m-11) (4751 m^5-49196 m^4+178555 m^3-265930 m^2+136104 m-3024) (k-m-1)]$\\
    9     &  $-\{(k-8) (k-2 m+7) [k^8 (10393 m^7-295169 m^6+3444763 m^5-21287315 m^4+74661412 m^3-147149156 m^2+148827072 m-58378320)+k^7 (93537 m^7-2629468 m^6+30398601 m^5-186195970 m^4+647532468 m^3-1265548552 m^2+1268906064 m-492972480)+k^6 (-463559 m^7+13467688 m^6-160270925 m^5+1007098102 m^4-3582455306 m^3+7142486032 m^2-7285730520 m+2869943328)+k^5 (-7040967 m^7+199285124 m^6-2316820929 m^5+14253708290 m^4-49732292718 m^3+97411690736 m^2-97808625096 m+38061051840)+k^4 (-14147654 m^7+380248249 m^6-4210420826 m^5+24757920421 m^4-82938771506 m^3+157002585316 m^2-153809636904 m+59074800624)+2 k^3 (39720207 m^7-1170499052 m^6+14047475388 m^5-88527503360 m^4+314151700473 m^3-621935009348 m^2+628262828172 m-245216321280)+12 k^2 (34497145 m^7-972446354 m^6+11227172479 m^5-68469189584 m^4+236767720390 m^3-460533277076 m^2+460767938016 m-179311977936)+72 k (9940368 m^7-263143441 m^6+2885073381 m^5-16942211875 m^4+57248478267 m^3-110112505504 m^2+109806816564 m-42781636800)+8640 (57016 m^7-1276772 m^6+12427702 m^5-68227355 m^4+223726804 m^3-426503693 m^2+425617338 m-166486320)]\}/[2 (k-7) (k-5) (k+2) (k+3) (k+4) (k+5) (k+6) (m-6)^2 (2 m-13) (10393 m^6-160060 m^5+931642 m^4-2537428 m^3+3195589 m^2-1476928 m+27720) (k-m-1)]$\\\hline
\end{longtable}

\clearpage

\section{Proof of the closed form for $D(n,p)$ from eq.~\eqref{diff_ex}}
\label{wzdiffproof}

In this appendix, we prove the closed-form expression for the difference function $D(n,p)$. In particular, we focus on the term that contains a summation over $k$.

\begin{lemma}
\label{wzdiff}
For all $(n,p)$ with $n\geq 5$ and $4 \leq p \leq n$, the following identity holds:
\begin{eqnarray}
F(n,p) & = & \sum_{k=p-1}^{n-1} \Bigg[ \frac{(k-p+2)^2 (k-p+4) (k+p-3)! ( 2n-p-k-1)!}{(k+1)! (p-3)! (n-k-1)! (n-p+1)!} \nonumber \\
    & & - \frac{(k-p+3)^2 (k-p+5) (k+p-4)! (2 n-p-k)!}{(k+1)! (p-4)! (n-k-1)! (n-p+2)!}\Bigg] \nonumber \\
& = & \frac{2p(np+7p-2p^2-6)(2n-2p+2)! \, (2p-5)!}{n(n-1)(n-p)! \,  (n-p+2)! \, (p-3)! \, p!}. \nonumber
\end{eqnarray}
\end{lemma}

\begin{proof}
Let $\Delta_k$ denote the forward difference operator in $k$, meaning that $\Delta_k(f) = f(k+1)-f(k)$. Let $f(n,p,k)$ be the summand in the expression for $F(n,p)$. We sum the equation
\begin{equation}
    \label{eq:appB_eq1}
    f(n,p,k) = \Delta_k\left(h_{n,p}(k)\right)
\end{equation}
over $k$, from $k=p-1$ to $k=n-1$.
The left-hand side of eq.~\eqref{eq:appB_eq1} is the summand in the statement of the lemma, and the function $h_{n,p}(k)$ is the output of Gosper's algorithm \citep{fastzeil, aeqb}:
\begin{eqnarray}
h_{n,p}(k) & = & \Big[(k+1) (2n-p-k) (-k^3 n^2+2 k^3 n-k^3+3 k^2 n^2 p-9 k^2 n^2-4 k^2 n p+15 k^2 n-3 k^2 p \nonumber \\ 
& & +6 k^2 -3 k n^2 p^2 + 18 k n^2 p-29 k n^2+2 k n p^2-6 k n p+4 k n-3 k p^2+12 k p-11 k \nonumber \\
& & +n^2 p^3-9 n^2 p^2 +25 n^2 p -21 n^2+3 n p^2-14 n p +15 n-p^3+6 p^2-11 p+6) \nonumber \\ \label{eqGosp}
& & (k+p-4)! (2 n-p-k-1)!\Big]/ \nonumber \\ 
& & \Big[(n-1) n (p-3)(k+1)! (n-p+2)(p-4)!(n-k-1)(n-p+1)!\Big]. 
\end{eqnarray}

We verify eq.~\eqref{eq:appB_eq1} by using eq.~\eqref{eqGosp}. After summation, the left-hand side becomes $F(n,p)$. The right-hand side telescopes, so that all terms except the first and the last cancel:
\begin{equation}
\label{eqhnpn}
    F(n,p) = h_{n,p}(n) - h_{n,p}(p-1).
\end{equation}
We obtain the statement of the lemma by algebraic simplification of the right-hand side of eq.~\eqref{eqhnpn}.
\end{proof}

\clearpage

\section{Proofs of inequalities required for the proof of 
Proposition \ref{prop:hmax_order}}

In the proof of Proposition \ref{prop:hmax_order}, we make use of lower and upper bounds for $h(n,p_m)$. First, we prove the lower bound.

\begin{lemma}
\label{lemma:h_lower_bound}
For all $n\geq3$, $h(n,p_m) \geq \mathcal{C}_{n-2}$.
\end{lemma}
\begin{proof}
From Theorem \ref{thm:maxp}, $h(n,p_m) \geq h(n,3)$. Comparing $h(n,3)$ from eq.~\eqref{eq_h3} and $\mathcal{C}_{n-2}$ from eq.~\eqref{catalan}, we obtain
$$\frac{h(n,3)}{\mathcal{C}_{n-2}} = \frac{3(n-2)}{n} \geq 1$$
for all $n \geq 3$. Hence, $\mathcal{C}_{n-2} \leq h(n,3) \leq h(n,p_m)$ as desired.
\end{proof}

To prove the upper bound, we first need an identity concerning Catalan numbers.
\begin{lemma} 
\label{lemma:decomposition}
For $n\geq 3$, the Catalan number $\mathcal{C}_n$ can be decomposed as a sum. \\
(i) For even $n \geq 4$, $n=2m+2$ for $m \geq 1$, 
    \begin{equation*}
        \mathcal{C}_{2m+2} = \sum_{k=m+2}^{2m+2} C(k,m+1) \, C(m,2m+2-k).
    \end{equation*}
(ii) For odd $n \geq 3$, $n=2m+1$ for $m \geq 1$, 
    \begin{equation*}
        \mathcal{C}_{2m+1} = \sum_{k=m+1}^{2m+1}  C(k,m) \, C(m,2m+1-k).
    \end{equation*}
\end{lemma}
\begin{proof}
(i) We use two ways of counting monotonic paths. $\mathcal{C}_{2m+2}$ gives the number of monotonic paths that travel from $(0,0)$ to $(2m+2,2m+2)$ on a square lattice, without crossing the diagonal connecting $(0,0)$ to $(2m+2,2m+2)$. Each of these paths passes through exactly one vertical edge from a point $(k,m+1)$ to a point $(k,m+2)$, where $k$ ranges from $m+2$ to $2m+2$.

The number of monotonic paths that travel from $(0,0)$ to $(k,m+1)$ and that do not cross the diagonal is $C(k,m+1)$. The number of monotonic paths from $(k,m+2)$ to $(2m+2,2m+2)$ that do not cross the diagonal is obtained by traversing the paths in reverse order, from $(2m+2,2m+2)$ down and to the left, reaching $(k,m+2)$ (Figure \ref{fig:upper}A). The associated number of paths is $C(m,2m+2-k)$. Hence, the total number of monotonic paths from $(0,0)$ to $(2m+2,2m+2)$ that do not cross the diagonal is $\sum_{k=m+2}^{2m+2} C(k,m+1) \, C(m,2m+2-k)$.

(ii) The argument in the odd case proceeds in the same way. Each path from $(0,0)$ to $(2m+1,2m+1)$ passes through exactly one vertical edge from a point $(k,m)$ to a point $(k,m+1)$, where $k$ ranges from $m+1$ to $2m+1$. The number of paths from $(0,0)$ to $(k,m)$ is $C(k,m)$, and the number of paths from $(k,m+1)$ to $(2m+1,2m+1)$ is $C(m,2m+1-k)$.
\end{proof}

We are now ready to prove the upper bound. 
\begin{lemma}
\label{lemma:h_upper_bound}
For all $n\geq 3$, $h(n,p_m) \leq  n\,\mathcal{C}_{n+2}$.
\end{lemma}
\begin{proof}
We split the proof into two cases, according to the expressions for $p_m$ from Theorem~\ref{thm:maxp}.

First, assume $n$ is even, $n = 2m$ for $m \geq 2$. Then $p_m = m+1$, and by Theorem \ref{thm:coalescent},
    \begin{equation*}
        h(2m,m+1) = \sum_{k=m}^{2m-1} (k-m+1) \, C(k,m-2) \, C(m-1,2m-k-1),
    \end{equation*}
or, equivalently,
    \begin{equation*}
        h(2m,m+1) = \sum_{k=m+2}^{2m+1} (k-m-1) \, C(k-2,m-2) \, C(m-1,2m-k+1).
    \end{equation*}

Using the decomposition in Lemma \ref{lemma:decomposition}, 
    \begin{eqnarray}
        2m\,\mathcal{C}_{2m+2}-h(2m,m+1) & = & \sum_{k=m+2}^{2m+1}\Big[2m\,C(k,m+1) \, C(m,2m-k+2)  \nonumber \\ & & -(k-m-1) \, C(k-2,m-2) \, C(m-1,2m-k+1)\Big]\nonumber\\
        & & + \frac{2m(m+2)(3m+3)!}{(m+1)!(2m+3)!}.\nonumber \label{eq:cat_h_difference_odd}
    \end{eqnarray}
    The summand in the first term is nonnegative, as function $C(n,k)$ is monotonically increasing with respect to both arguments, and $2m \geq k-m-1$ because $k \leq 2m+1$. The remaining term is also nonnegative. Hence, for even $n$ we indeed have $h(n,p_m) \leq n\,\mathcal{C}_{n+2}$.
    
Now assume $n$ is odd, with $n = 2m-1$ and $m \geq 2$. We then must show $(2m-1)\,\mathcal{C}_{2m+1} \geq h(2m-1,m+1)$. By Theorem \ref{thm:coalescent}, we have 
    \begin{equation*}
        h(2m-1,m+1) = \sum_{k=m}^{2m-2} (k-m+1)\,C(k,m-2)\,C(m-2,2m-k-2),
    \end{equation*}
    or, equivalently,
    \begin{equation*}
        h(2m-1,m+1) = \sum_{k=m+1}^{2m-1} (k-m) \, C(k-1,m-2) \, C(m-2,2m-k-1).
    \end{equation*}
    
Using the decomposition in Lemma \ref{lemma:decomposition}, we have
    \begin{eqnarray}
        (2m-1)\,\mathcal{C}_{2m+1}-h(2m-1,m+1) & = & \sum_{k=m+1}^{2m-1}\Big[(2m-1)\,C(k,m)\,C(m,2m-k+1) \nonumber \\ 
        & & -(k-m)\,C(k-1,m-2)\,C(m-2,2m-k-1)\Big]\nonumber\\
        & & + \frac{(2m-1)\left(2 m^3+7 m^2+9 m+2\right)(3m)!}{m!(2m+2)!}.\nonumber \label{eq:cat_h_difference_even}
    \end{eqnarray}
As is true in the even case, the summand is termwise nonnegative, as is the remaining term. We conclude that for odd $n$, $h(n,p_m) \leq n\,\mathcal{C}_{n+2}$, completing the proof.
\end{proof}

\clearpage
\section{Proof of the closed form for $D^*(n,p)$ from eq.~\eqref{dualdiff_ex}}
\label{dualdiffproof}

Here we prove the closed-form expression for the dual difference function in eq.~\eqref{dualdiff_ex} from Section~\ref{secDualDifference}.

\begin{lemma}
For all $(n,p)$ with $n\geq 5$ and $4 \leq p \leq n$, the following identity holds:
\begin{eqnarray}
F(n,p) & = & \sum_{k=n-p+3}^{n-1} 
   \frac{(k-n+p-1)^2 (k-n+p+1) (k+n-p)! (n+p-k-4)!}{(k+1)! (p-2)! (n-k-1)! (n-p)!}\Bigg] \nonumber \\
	 & & -   \frac{(k-n+p-2)^2 (k-n+p) (k+n-p+1)! (n+p-k-5)!}{(k+1)! (p-3)! (n-k-1)! (n-p+1)!}\Bigg] \nonumber \\
	 & = &  \frac{4 (p-3)(n^2+2p^2-3np+5n-9p+10) (2 p-7)!(2n-2p+3)!}{n(n-1) (n-p+1)(p-4)!(p-2)!(n-p)!(n-p+3)!}.\nonumber 
\end{eqnarray}
\end{lemma}
\begin{proof}
As in Appendix B, let $\Delta_k$ denote the forward difference operator in $k$, meaning that 
$\Delta_k(f) = f(k+1)-f(k)$. Let $f(n,p,k)$ be the summand in the expression for $F(n,p)$.

We sum the equation
\begin{equation}
    \label{eq:appD_eq1}
    f(n,p,k) = \Delta_k\left(h_{n,p}(k)\right)
\end{equation}
over $k$, from $k=n-p+3$ to $k=n-1$. The left-hand side is the summand in the statement of the lemma, and the function $h_{n,p}(k)$ is the output of Gosper's algorithm \citep{fastzeil, aeqb}:
\begin{eqnarray}
& h_{n,p}(k)=  &  \Big[(k+1)(k^3 n^2-2 k^3 n+k^3-3 k^2 n^3+3 k^2 n^2 p+k^2 n^2-4 k^2 n p +  4 k^2 n-3 k^2 p \nonumber \\ 
          & & +  6 k^2+3 k n^4-6 k n^3 p+4 k n^3+3 k n^2 p^2-2 k n^2 p -  2 k n^2-2 k n p^2+4 k n p + 3 k p^2 \nonumber \\ 
          & & -  12 k p+11 k-n^5+3 n^4 p-3 n^4-3 n^3 p^2+6 n^3 p - 3 n^3+n^2 p^3-3 n^2 p^2+4 n^2 p \nonumber \\ 
          & & -  3 n^2-2 n p+4 n-p^3+6 p^2-11 p+6) (n+k-p)! (n+p-k-4)!\Big]/ \nonumber \\ \label{eqGosper}
          & &   \Big[n (n-1) (k+1)! (p-2)! (n-k-1)! (n-p+1)!\Big].
\end{eqnarray}

With $h$ as in eq.~(\ref{eqGosper}),  eq.~\eqref{eq:appD_eq1} is verified algebraically. After summation of eq.~\eqref{eq:appD_eq1}, the left-hand side becomes $F(n,p)$, and the right-hand side telescopes. All terms except the first and the last cancel, leaving 
\begin{equation}
\label{eqTelescope}
    F(n,p) = h_{n,p}(n) - h_{n,p}(n-p+3).
\end{equation}
The lemma then follows by algebraic simplification of the right-hand side of eq.~\eqref{eqTelescope}.
\end{proof}

\end{document}